\pdfoutput=1

\documentclass[acmsmall,screen,authorversion,acmthm]{acmart}
\usepackage{caption}
\usepackage{subcaption}
\usepackage{prooftree}
\usepackage{soul}
\usepackage{tikz}
\usepackage{siunitx}
\usepackage{listings}
\usepackage{pifont}
\usepackage{multirow}
\usepackage{amsmath}
\usepackage{amsfonts}
\usepackage{amssymb}
\usepackage{graphicx}
\usepackage{paralist}
\usepackage{wrapfig}
\usepackage{balance}
\usepackage{dsfont}
\usepackage{enumitem}
\usepackage{mathtools}
\usepackage{listings}
\usepackage[plain]{algorithm}
\usepackage[noend]{algpseudocode}
\usepackage{pifont}
\usepackage{stmaryrd}
\usepackage{blindtext}
\usepackage{enumitem}
\usepackage{comment}
\usepackage{url}
\usepackage{etoolbox}
\usepackage{pbox}
\usepackage{csquotes}
\usepackage{ctable} % for \specialrule command
\usepackage{textcase}
\usepackage{wasysym}
\usepackage{mdframed}
\usepackage{amsthm}
\usepackage{cleveref}

\renewcommand{\paragraph}[1]{\vspace{.03in}\noindent\textbf{{#1}~~~}}
\renewcommand{\qedsymbol}{{\footnotesize{$\blacksquare$}}}
\AtEndEnvironment{example}{\hfill\qedsymbol}%
\AtEndEnvironment{definition}{\hfill\qedsymbol}%

\usepackage{enumitem}% http://ctan.org/pkg/enumitem
\makeatletter
\lst@InstallKeywords k{attributes}{attributestyle}\slshape{attributestyle}{}ld
\makeatother

\usepackage[scaled=0.80]{DejaVuSansMono}
\lstdefinestyle{customc}{
  belowcaptionskip=1\baselineskip,
  breaklines=true,
  xleftmargin=\parindent,
  language=C,
  showstringspaces=false,
  basicstyle=\ttfamily,
  keywordstyle=\bfseries\ttfamily\color{keywordcolor},
  commentstyle=\itshape\color{black},
  identifierstyle=\ttfamily\color{black},
  stringstyle=\itshape\color{NavyBlue},
  keywords={ map, flatMap, reduce, then, in, if, else, reduceByKey},
moreattributes={let, where}, % etc...
attributestyle = \bfseries\ttfamily\color{attributecolor}
}

\setlist[description]{
   labelindent=.3cm,
   style=unboxed,
   leftmargin=.3cm,
   topsep=2pt,
   itemsep=1ex
}

\newcommand{\aws}[1]{{\color{brown}{\noindent\textsf{\textbf{Aws}---#1}}}}
 \newcommand{\jh}[1]{{\color{blue}{\noindent\textsf{\textbf{Justin}---#1}}}}
 \newcommand{\note}[1]{{\color{magenta}{\noindent\textsf{\textbf{Note}---#1}}}}
 \definecolor{mypink3}{cmyk}{0, 0.7808, 0.4429, 0.1412}
\newcommand{\maybe}[1]{{\color{mypink3} \textsf{\textbf{Remove?}---#1}}}

 \renewcommand{\aws}[1]{}
 \renewcommand{\jh}[1]{}
 \renewcommand{\note}[1]{}
 \renewcommand{\maybe}[1]{}

\everypar{\looseness=-1}

\definecolor{lightgray}{gray}{0.9}
\definecolor{midgray}{gray}{0.65}
\definecolor{darkgray}{gray}{0.4}

\newcommand{\abr}[1]{\textsc{\MakeLowercase{#1}}}

\newcommand{\abrs}[1]{\abr{#1}{\footnotesize{s}}\xspace}

\renewcommand{\vec}[1]{\boldsymbol{#1}}

\renewcommand{\leq}{\leqslant}
\renewcommand{\geq}{\geqslant}

\newcommand{\rone}{(\emph{i})~}
\newcommand{\rtwo}{(\emph{ii})~}
\newcommand{\rthree}{(\emph{iii})~}
\newcommand{\rfour}{(\emph{iv})~}

\definecolor{keywordcolor}{gray}{0.0}
\definecolor{attributecolor}{gray}{0.0}
\definecolor{wildcolor}{gray}{0.8}

\usepackage{relsize}

\newcommand\earr\hookrightarrow
\newcommand{\wild}{{\scriptsize \color{black} \CIRCLE}}

\let\oldwild\wild
\renewcommand{\wild}{%
  \mathchoice{\raisebox{.4pt}{$\displaystyle\oldwild$}}
             {\raisebox{.5pt}{$\oldwild$}}
             {\raisebox{0.5pt}{$\scriptstyle\oldwild$}}
             {\raisebox{0.2pt}{$\scriptscriptstyle\oldwild$}}}

\newcommand{\dpriv}{\abr{DP}\xspace}
\newcommand{\edp}{$\epsilon$-\dpriv}

\newcommand{\aset}{B}
\newcommand{\aelem}{b}

\newcommand{\dist}{\mathit{dist}}
\newcommand{\sdist}{\mathit{dist}_\downarrow}
\newcommand{\supp}{\mathit{supp}}

\newcommand{\aspace}[1]{\leftrightsquigarrow^{#1}}

\newcommand{\pr}[1]{\mathrm{Pr}[#1]}

\newcommand{\cpred}{\Pi}
\newcommand{\prog}{P}
\newcommand{\vars}{V}
\newcommand{\states}{S}
\newcommand{\varsi}{\vars^i}
\newcommand{\varsl}{\vars^l}

\newcommand{\locs}{L}
\newcommand{\lentry}{{\ell_\emph{en}}}
\newcommand{\lexit}{{\ell_\emph{ret}}}

\newcommand{\edges}{E}
\newcommand{\retv}{{v_r}}
\newcommand{\retva}{{v_{r_1}}}
\newcommand{\retvb}{{v_{r_2}}}

\newcommand{\stmt}{\mathit{s\!t}}
\newcommand{\stmts}{\emph{Stmts}}

\newcommand{\expr}{\mathit{exp}}
\newcommand{\bexpr}{\mathit{bexp}}
\newcommand{\dexpr}{\mathit{dexp}}
\newcommand{\iexpr}{{\mathit{iexp}}}

\newcommand{\lap}{\mathsf{Lap}}
\newcommand{\olap}{\mathsf{Exp}}

\newcommand{\assume}{\mathsf{assume}}

\newcommand{\sem}[1]{\llbracket #1\rrbracket}

\newcommand{\adj}{\Delta}

\newcommand{\dom}{D}

\newcommand{\trace}{\sigma}
\newcommand{\traces}{\Sigma}

\newcommand{\cpost}{\mathit{post}}

\newcommand{\eps}\epsilon

\newcommand{\domain}[1]{\mathit{dom}(#1)}

\newcommand{\cinit}{\textsc{init}}
\newcommand{\cassign}{\textsc{assign}}
\newcommand{\casmconst}{\textsc{assume}}
\newcommand{\casmsync}{\textsc{assume-s}}
\newcommand{\ccouple}{\textsc{couple}}
\newcommand{\cstrategy}{\textsc{strat}}
\newcommand{\cdp}{\textsc{dpriv}}

\newcommand{\strategy}{\mathit{strat}}
\newcommand{\clss}{\mathcal{C}}
\newcommand{\cls}{C}
\newcommand{\inv}{\emph{inv}}
\newcommand{\ret}{\iota}

\newcommand{\distance}{\theta}
\newcommand{\cvars}{\vec{v}}
\newcommand{\costvar}{\omega}

\newcommand{\rels}{R}
\newcommand{\rel}{r}
\newcommand{\interp}{\rho}
\newcommand{\hmc}{\abr{HMC}\xspace}
\newcommand{\enc}[1]{#1}

\newcommand{\head}{H}

\newcommand{\strp}{\tau}

\newcommand{\fs}{F}
\newcommand{\f}{f}

\algrenewcommand\algorithmicfunction{\textsf{\textbf{fun}}}

\newcommand{\dunit}{\mathit{dunit}}
\newcommand{\dbind}{\mathit{dbind}}
\newcommand{\dzero}{\mathit{dzero}}

\newcommand{\pre}{\emph{pre}}
\newcommand{\true}{\mathit{true}}
\newcommand{\trans}{\textsc{enc}}

\renewcommand{\epsilon}{\varepsilon}

\usepackage{cleveref}

\crefname{section}{\S}{\S}
\Crefname{section}{\S}{\S}
\crefname{figure}{Fig.}{Fig.}
\crefname{table}{Table}{Table}
\crefname{definition}{Def.}{Def.}
\crefname{theorem}{Thm.}{Thm.}
\crefname{lemma}{Lem.}{Lem.}
\crefname{example}{Ex.}{Ex.}

\usepackage{booktabs}   %% For formal tables:
\usepackage{subcaption} %% For complex figures with subfigures/subcaptions
\newtoggle{long}
\toggletrue{long}

\makeatletter\if@ACM@journal\makeatother
\acmJournal{PACMPL}
\acmVolume{2}
\acmNumber{POPL}
\acmArticle{58}
\acmYear{2018}
\acmMonth{1}
\acmDOI{10.1145/3158146}
\startPage{1}
\else\makeatother
\acmConference[PL'17]{ACM SIGPLAN Conference on Programming Languages}{January 01--03, 2017}{New York, NY, USA}
\acmYear{2017}
\acmISBN{978-x-xxxx-xxxx-x/YY/MM}
\acmDOI{10.1145/nnnnnnn.nnnnnnn}
\startPage{1}
\fi

\setcopyright{rightsretained}
\begin{document}

%% Title information
\title{Synthesizing Coupling Proofs of Differential Privacy}         %% [Short Title] is optional;
                                        %% when present, will be used in
                                        %% header instead of Full Title.
\iftoggle{long}{}{\titlenote{The extended version is available at \url{https://arxiv.org/abs/1709.05361}.}}
%\subtitle{A Differential Privacy Perspective}                     %% \subtitle is optional
%\subtitlenote{with subtitle note}       %% \subtitlenote is optional;
                                        %% can be repeated if necessary;
                                        %% contents suppressed with 'anonymous'

%% Author information
%% Contents and number of authors suppressed with 'anonymous'.
%% Each author should be introduced by \author, followed by
%% \authornote (optional), \orcid (optional), \affiliation, and
%% \email.
%% An author may have multiple affiliations and/or emails; repeat the
%% appropriate command.
%% Many elements are not rendered, but should be provided for metadata
%% extraction tools.

%% Author with single affiliation.
\author{Aws Albarghouthi}
%\authornote{with author1 note}          %% \authornote is optional;
                                        %% can be repeated if necessary
%\orcid{nnnn-nnnn-nnnn-nnnn}             %% \orcid is optional
\affiliation{%
  \institution{University of Wisconsin--Madison}            %% \institution is required
  \department{Computer Sciences Department}
  \streetaddress{1210 West Dayston St.}
  \city{Madison}
  \state{WI}
  \postcode{53706}
  \country{USA}}
\email{aws@cs.wisc.edu}          %% \email is recommended

%% Author with two affiliations and emails.
\author{Justin Hsu}
%\authornote{with author2 note}          %% \authornote is optional;
                                        %% can be repeated if necessary
\orcid{0000-0002-8953-7060}             %% \orcid is optional
\affiliation{%
  \institution{University College London}           %% \institution is required
  \city{London}
  \country{UK}}
\email{email@justinh.su}         %% \email is recommended

%% Paper note
%% The \thanks command may be used to create a "paper note" ---
%% similar to a title note or an author note, but not explicitly
%% associated with a particular element.  It will appear immediately
%% above the permission/copyright statement.
%\thanks{with paper note}                %% \thanks is optional
                                        %% can be repeated if necesary
                                        %% contents suppressed with 'anonymous'

%% Abstract
%% Note: \begin{abstract}...\end{abstract} environment must come
%% before \maketitle command
\begin{abstract}
\emph{Differential privacy} has emerged as
a promising probabilistic formulation of privacy, generating
intense interest within academia and industry.
We present a push-button, automated
 technique for verifying $\eps$-differential privacy
of sophisticated randomized algorithms.
We make several conceptual, algorithmic,
and practical contributions:
\rone Inspired by the recent advances on \emph{approximate couplings} and
\emph{randomness alignment},
we present a new proof technique called \emph{coupling strategies},
which casts differential privacy proofs as a winning strategy
in a game where we have finite privacy resources to expend.
\rtwo To discover a \emph{winning strategy},
we present a constraint-based formulation of the problem
as a set of \emph{Horn modulo couplings} (\hmc) constraints,
a novel combination of
first-order Horn clauses and probabilistic constraints.
\rthree We present a technique for solving \hmc constraints by
transforming probabilistic constraints into logical constraints
with uninterpreted functions.
\rfour Finally, we implement our technique in the FairSquare verifier and provide
the first automated privacy proofs for a number of challenging algorithms
from the differential privacy literature, including
Report Noisy Max, the Exponential Mechanism,
and the Sparse Vector Mechanism.
\end{abstract}

\begin{CCSXML}
<ccs2012>
<concept>
<concept_id>10002978.10002986.10002990</concept_id>
<concept_desc>Security and privacy~Logic and verification</concept_desc>
<concept_significance>500</concept_significance>
</concept>
<concept>
<concept_id>10003752.10003790.10003806</concept_id>
<concept_desc>Theory of computation~Programming logic</concept_desc>
<concept_significance>500</concept_significance>
</concept>
</ccs2012>
\end{CCSXML}


% Keywords
% comma separated list
  %% \keywords is optional

%% \maketitle
%% Note: \maketitle command must come after title commands, author
%% commands, abstract environment, Computing Classification System
%% environment and commands, and keywords command.
\maketitle

%\iftoggle{long}{\textbf{\emph{This is a long version with proofs and further details in the Appendix}}}{}

%!TEX root=paper.tex

\section{Introduction} \label{sec:intro}

% \note{fix comparison with LightDP}
% \jh{Fixed?}
As more and more personal information is aggregated into massive databases, a
central question is how to safely use this data while protecting privacy.
\emph{Differential privacy}~\citep{DMNS06} has recently emerged as
one of the most promising formalizations of privacy. Informally, a randomized
program is differentially private if on any two input databases differing in a
single person's private data, the program's output distributions
are \emph{almost} the same; intuitively, a private program shouldn't
depend (or reveal) too much on any single individual's record. Differential privacy
models privacy quantitatively, by bounding how much the output distribution can
change.

Besides generating intense interest in fields like machine learning, theoretical
computer science, and security, differential privacy has proven to be a
surprisingly fruitful
target for formal verification. By now, several verification techniques can
\emph{automatically} prove privacy given a lightly annotated program; examples
include linear type systems~\citep{ReedPierce10,GHHNP13} and various flavors of
dependent types~\citep{BGGHRS15,zhang2016autopriv}.
For verification, the key feature of differential privacy is \emph{composition}:
private computations can be combined into more complex algorithms while
automatically satisfying a similar differential privacy guarantee.
Composition properties make it easy to design new private algorithms, and also
make privacy feasible to verify.

While composition is a powerful tool for proving privacy, it often
gives an overly conservative estimate of the privacy level for more sophisticated algorithms.
Verification---in particular, automated verification---has been far more
challenging for these examples. However, these algorithms are highly important
to verify since their proofs are quite subtle. One example, the \emph{Sparse
  Vector} mechanism~\citep{dwork2009complexity}, has been proposed in slightly different
forms at least six separate times throughout the literature. While there were
proofs of privacy for each variant, researchers later discovered that only two
of the proofs were correct \citep{lyu2016understanding}. Another example, the
\emph{Report Noisy Max} mechanism~\citep{dwork2014algorithmic}, has also suffered
from flawed privacy proofs in the past.

\subsection{State-of-the-art in Differential Privacy Verification}
Recently, researchers have developed techniques to verify private algorithms
beyond composition: \emph{approximate couplings}~\citep{BGGHS16} and
\emph{randomness alignment}~\citep{zhang2016autopriv}.

\begin{description}
  \item[Approximate Couplings.]
  Approximate couplings are a generalization of \emph{couplings}~\citep{lindvall2002lectures} from probability
  theory. Informally, a coupling models two distributions with a single
  correlated distribution, while an approximate coupling approximately models the
  two given distributions. Given two output distributions of a
  program, the existence of a coupling with certain properties can imply target
  properties about the two output distributions, and about the program itself.
  % Crucially
  % for formal reasoning, couplings can be constructed compositionally: we can build
  % a coupling of two output distributions by selecting a coupling for each basic
  % sampling instruction in the program. Proving a probabilistic property, like
  % differential privacy, boils down to selecting appropriate couplings to ensure
  % that the final coupling has a particular form.

  Approximate couplings are a rich
  abstraction for reasoning about differential privacy, supporting clean,
  compositional proofs for many examples beyond the reach of the
  standard composition theorems for differential privacy.
  \citet{BGGHS16} first explored approximate couplings for proving differential
  privacy, building approximate couplings in a relational program logic
  apRHL where the rule for the random sampling command selects an
  approximate coupling.
  While this approach is quite expressive, there are two major drawbacks:
  \rone Constructing proofs requires significant manual effort---existing
  proofs are formalized in an interactive proof assistant---and selecting the
  correct couplings requires considerable ingenuity.
  \rtwo Like Hoare logic, apRHL does not immediately lend itself to a systematic
  algorithmic technique for finding invariants,
  making it a challenging target for automation.
  \item[Randomness Alignment.]
  In an independent line of work, \citet{zhang2016autopriv} used a technique called
  \emph{randomness alignment} to prove differential privacy. Informally, a
  randomness alignment for two distributions is an injection pairing the samples
  in the first distribution with samples in the second distribution, such that
  paired samples have approximately the same probability.
  % By selecting randomness
  % alignments for each sampling command and flowing them through the program, if
  % the final alignment always pairs equal results, then the program is
  % differentially private.
  \citet{zhang2016autopriv} implemented their approach in a semi-automated
  system called LightDP, combining a dependent type system,  a custom
  type-inference algorithm to search for the desired randomness alignments, and a
  product program
  construction. Notably, LightDP can analyze standard examples as well as more
  advanced examples like the Sparse Vector mechanism.  While LightDP is an
  impressive achievement, it, too has a few shortcomings:
  \rone  While randomness alignments can be inferred
  by LightDP's sophisticated type inference algorithm, the full analysis depends
  on a product program construction with manually-provided loop invariants.
  \rtwo Certain examples provable by approximate couplings seem to lie beyond
  the reach of LightDP.
  \rthree Finally, the analysis in LightDP is technically complex,
  involving three kinds of program analysis.
\end{description}

While the approaches seem broadly similar, their precise relation remains unclear.
Proofs using randomness alignment are often simpler---critical for
automation---while approximate couplings seem to be a more expressive proof
technique.

\subsection{Our Approach: Automated Synthesis of Coupling Proofs}
We aim to blend the best features of both approaches---expressivity and
automation---enabling fully automatic construction of coupling proofs. To do so, we
\rone present a novel formulation of proofs of differential privacy via couplings,
\rtwo present a constraint-based technique for discovering coupling proofs, and
\rthree discuss how to solve these constraints using techniques
from automated program synthesis and verification.
We describe the key components and novelties of our approach below, and
illustrate the overall process in \cref{fig:overview}.

\begin{figure}[t]
\includegraphics[scale=1.15]{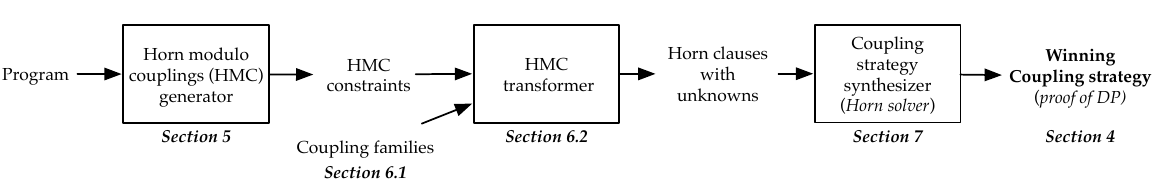}
\vspace{-.1in}
\caption{Overview of our coupling strategy synthesis methodology}
\label{fig:overview}
%\vspace{-.3in}
\end{figure}

\paragraph{Proofs via Coupling Strategies.}
The most challenging part of a coupling proof is selecting an appropriate
coupling for each sampling instruction. Our first insight is that we can
formalize the proof technique as discovering a \emph{coupling strategy}, which
chooses a coupling for each part of the program while ensuring that the
couplings can be composed together. A \emph{winning} coupling strategy proves
that the program is differentially private.

To show soundness, we propose a novel, fine-grained version of approximate
couplings---\emph{variable approximate couplings}---where the privacy level can
vary depending on the pair of states. We show that a winning coupling strategy
encodes a variable approximate coupling of the two output distributions, establishing
differential privacy for the program. To encode more complex proofs, we
give new constructions of variable approximate couplings inspired by the
randomness alignments of \citet{zhang2016autopriv}.

\paragraph{Horn Modulo Couplings.}
To automatically synthesize coupling strategies, we describe
winning coupling strategies as a set of constraints. Constraint-based techniques
are a well-studied tool for
synthesizing loop invariants, rely-guarantee proofs, ranking
functions, etc. For instance, \emph{constrained Horn clauses} are routinely
used as a first-order-logic representation of proof rules~\citep{bjorner2015horn,grebenshchikov2012}.
The constraint-based style of verification cleanly decouples the
theoretical task of devising the proof rules from the algorithmic task of
solving the constraints.
% by now, there are many mature solvers for this latter task (e.g.,
% \abr{HSF}~\citep{grebenshchikov2012hsf}, SeaHorn~\citep{gurfinkel2015seahorn},
% and engines in the popular Z3 \abr{SMT}
% solver~\citep{hoder2011muz,mcmillan2013solving}).

However, existing constraint systems are largely geared towards deterministic
programs---it is not clear how to encode randomized algorithms and the
differential privacy property.
We present a novel system of constraints
called \emph{Horn modulo couplings} (\hmc), extending Horn clauses with
probabilistic \emph{coupling constraints} that use first-order relations to
encode approximate couplings. \hmc constraints can describe
winning coupling strategies.

\paragraph{Solving Horn Modulo Couplings Constraints.}
Solving \hmc constraints is quite challenging, as they combine
logical and probabilistic constraints. To simplify the task, we
transform probabilistic coupling constraints into logical constraints with
unknown expressions. The key idea is that we
can restrict coupling strategies to use known approximate
couplings from the privacy literature, augmented with a few new constructions we
propose in this work.

Our transformation yields a constraint of the form $\exists f \ldotp \forall x
\ldotp \phi$, read as: there exists a strategy $f$ such that for all inputs $x$,
the program is differentially private. We employ established techniques from
synthesis and program verification to solve these simplified constraints.
Specifically, we use \emph{counterexample-guided inductive synthesis}
(\abr{CEGIS})~\citep{DBLP:conf/asplos/Solar-LezamaTBSS06} to discover a strategy
$f$, and \emph{predicate abstraction}~\citep{graf1997construction} to prove that
$f$ is a winning strategy.

\paragraph{Implementation and Evaluation.}
We have implemented our technique in the FairSquare probabilistic verification
infrastructure~\citep{albarghouthi17} and used it to automatically prove
$\epsilon$-differential privacy of a number of algorithms from the literature
that have so far eluded automated verification. For example, we give the first fully
automated proofs of Report Noisy Max~\citep{dwork2014algorithmic}, the discrete
exponential mechanism~\citep{MT07}, and
the Sparse Vector mechanism~\citep{lyu2016understanding}.

\subsection{Outline and Contributions}

After illustrating our technique on a motivating example
(\cref{sec:example}), we present our main  contributions.
\begin{itemize}
  \item We introduce \emph{coupling strategies} as a way of representing
    coupling proofs of $\eps$-differential privacy (\cref{sec:strategies}). To
    prove soundness, we develop a novel generalization of approximate couplings
    to support more precise reasoning about the privacy level (\cref{sec:prelims}).

  \item We introduce \emph{Horn modulo couplings}, which enrich first-order
    Horn clauses with probabilistic coupling constraints.  We show how to reduce
    the problem of proving differential privacy to solving Horn modulo
    couplings constraints (\cref{sec:hcc}).

  \item We show how to automatically solve Horn modulo couplings constraints by
    transforming probabilistic coupling constraints into logical constraints
    with unknown expressions (\cref{sec:couplings}).

  \item We use automated program synthesis and verification techniques to
    implement our system, and demonstrate its ability to efficiently
    and automatically prove privacy of a number of differentially private
    algorithms from the literature (\cref{sec:evaluation}).
\end{itemize}
Our approach marries ideas from probabilistic relational logics
and type systems with constraint-based verification and synthesis; we compare
with other verification techniques (\cref{sec:related}).

% Overall, our work gives the first system to automate coupling proofs.  While we
% focus on $\eps$-differential privacy, we lay the groundwork towards other
% targets of coupling proofs, including \emph{convergence}, \emph{uniformity},
% \emph{independence}, and more. We conclude by discussing these and other
% promising future directions for automating proofs for probabilistic programs
% (\cref{sec:conclusion}).

\section{Motivation and Illustration}\label{sec:example}

\begin{wrapfigure}[17]{c}{.32\linewidth}
  \vspace{-.18in}
\begin{subfigure}[c]{0.30\textwidth}
  \centering
  \footnotesize
\hspace{.2in}\includegraphics[scale=1.3]{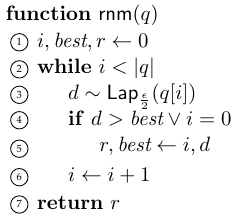}
\end{subfigure}

\vspace{.15in}

\begin{subfigure}[c]{0.30\textwidth}
  \footnotesize
\hspace{.2in}\includegraphics[scale=.15]{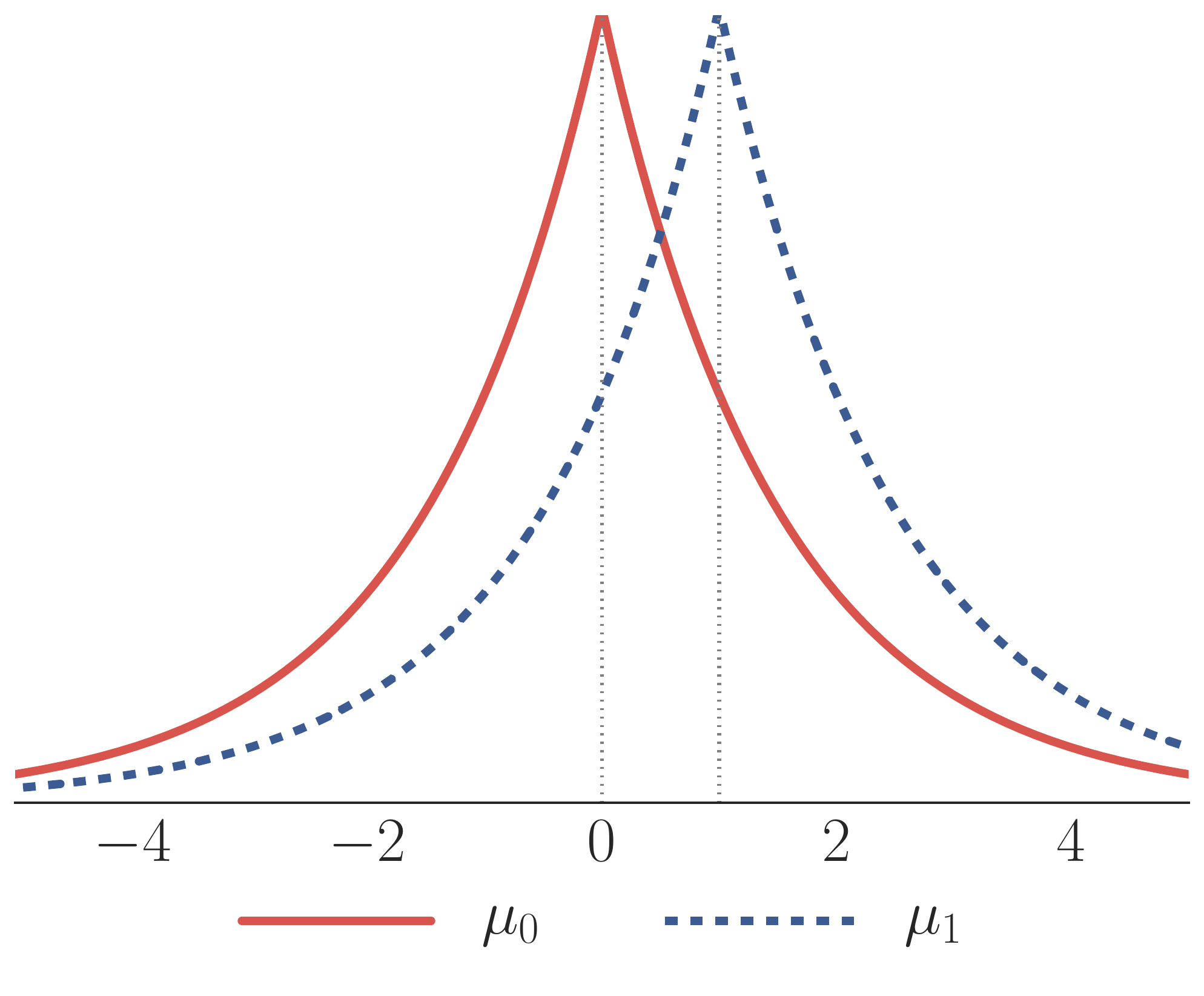}
%\caption{Probability density functions of
%two Laplace distributions, $\mu_i$,
%with mean $i$}
\end{subfigure}
\vspace{-.1\in}
\caption{(top) Report Noisy Max.
  (bot) Laplace distrs.}\label{fig:ex}
\end{wrapfigure}

In this section, we demonstrate our verification technique on a
concrete example.
% First, we informally describe how to prove differential
% privacy using \emph{coupling strategies} (\cref{ssec:coupling}).  Then,
% we demonstrate how to automatically synthesize coupling strategies in two
% phases: expressing the task as a system of recursive Horn clauses with
% \emph{coupling constraints}, and solving it by
% translating coupling constraints
% into a set of Horn clauses with unknowns (\cref{ssec:details}).
%
Our running example is an algorithm from the differential privacy literature
called \emph{Report Noisy Max}~\citep{dwork2014algorithmic}, which finds the
query with the highest value in a list of counting queries.  For instance, given
a database of medical histories, each query in the list could count the
occurrences of a certain medical condition.  Then, Report Noisy Max would reveal
the condition with approximately the most occurrences, adding random noise to
protect the privacy of patients.

We implement this algorithm in
\cref{fig:ex}~(top) as the program \textsf{rnm}. Given an array $q$ of numeric
queries of size $|q|$, \textsf{rnm} evaluates each query, adds noise to the
answer, and reports the index of the query with the largest noisy answer.
To
achieve $\epsilon$-differential privacy (\edp, for short), \emph{Laplacian
  noise} is added to the result of each query in the first statement of the
while loop. (\cref{fig:ex} (bottom) shows two Laplace distributions.)
% Report Noisy Max is typically used to find the query with the highest
% value in a set of $|q|$ counting queries.  For instance,
% given a database of medical histories,
% a query $q[i]$ could count the number of occurrences
% of a medical condition $i$.
% Then, Report Noisy Max would reveal the
% condition with approximately the most occurrences,
% adding random noise to protect the
% privacy of patients.
While the code is simple, proving \edp is anything but---the only existing
formal proof uses the probabilistic relational program logic apRHL and was
carried out in an interactive proof assistant \citep{BGGHS16}.

\subsection{Proof by approximate coupling, intuitively} \label{ssec:coupling}
Differential privacy is a \emph{relational} property of programs (also known as
a \emph{hyperproperty}) that compares the output distributions on two different
inputs, modeling two hypothetical versions of the private data that should be considered
indistinguishable. In the case of \textsf{rnm}, the private data is simply the
list of query answers $q[i]$; we will assume that two input states $(s_1, s_2)$
are similar (also called \emph{adjacent}) when each query answer differs by at
most $1$ in the two inputs; formally,  $\forall j \in [0,n) \ldotp |q_1[j] - q_2[j]|
\leq 1$, where $q_1,q_2$ are copies of $q$ in the adjacent states,
and $|q_1|=|q_2|=n$.

Adding to the verification challenge,
differential privacy is also a \emph{quantitative} property, often parameterized
by a number $\epsilon$.  To prove $\epsilon$-differential privacy, we must show
that the probability of any final output is approximately the same in \emph{two
  executions} of the program from any two adjacent inputs. The
degree of similarity---and the strength of the privacy guarantee---is governed
by $\epsilon$: smaller values of $\epsilon$ guarantee more similar
probabilities, and yields stronger privacy guarantees.  Formally, \textsf{rnm} is
\edp if
\begin{displayquote}
  \emph{for every pair of adjacent inputs $q_1$
    and $q_2$, every possible output $j$, and  every $\eps > 0$, }
  \[
    \pr{\textsf{rnm}(q_1) = j} \leq \exp(\eps) \cdot \pr{\textsf{rnm}(q_2) = j}
  \]
\end{displayquote}
While it could theoretically be possible to analyze the two executions
separately and then verify the inequality, this approach is highly complex.
Instead, let us imagine that
we step through the two executions side-by-side while tracking the two states.
Initially, we have the input states $(s_1, s_2)$. A deterministic instruction
simply updates both states, say to $(s_1', s_2')$.

Random sampling instructions $x \sim \mu$ are more challenging to handle.  Since
we are considering two sampling instructions, in principle we need to consider
all possible pairs of samples in the two programs. However, since we eventually
want to show that the probability of some output in the first run is close to
the probability of that same output in the second run, we can imagine
\emph{pairing} each result from the first sampling instruction with a
corresponding result from the second sampling instruction, with approximately
the same probabilities. This pairing yields a set of paired
states; intuitively, one for each possible sample. By flowing this set
forward through the program---applying deterministic instructions, selecting how
to match samples for sampling instructions, and so on---we end up pairing every
possible output state in the first execution with a corresponding output state
in the second execution with approximately the same probability.

Proving differential privacy, then, boils down to cleverly finding a pairing for
each sampling instruction in order to ensure that all paired output states are
related in some particular way. For instance, if whenever the first output state
has return value $5$ its matching state also has return value $5$, then the
probability of returning $5$ is roughly the same in both distributions.

\paragraph{Approximate Couplings.}
% \aws{i removed mentions of bijection -- i think it overcomplicates things at this point in the paper}
This idea for proving differential privacy can be formalized as a \emph{proof by
  approximate coupling}, inspired by the \emph{proof by coupling} technique from
probability theory (see, e.g., \citet{lindvall2002lectures}). An approximate
coupling for two probability distributions $\mu_1, \mu_2$ is a relation $\Lambda \subseteq
\states \times \states \times \mathds{R}$ attaching a non-negative real number
$c$ to every pair of linked (or \emph{coupled}) states $(s_1, s_2)$.
%Approximate
%couplings are simplest to understand when $\Lambda$ gives a bijection of states.
The parameter $c$ bounds how far apart the probability of $s_1$ in
$\mu_1$ is from the probability of $s_2$ in $\mu_2$: when $c = 0$ the two
probabilities must be equal, while a larger $c$ allows greater differences in
probabilities.
%When $\Lambda$ is not a bijection it is more challenging to
%interpret approximate couplings, but the basic intuition of linking two events
%with approximately equal probabilities in $\mu_1$ and $\mu_2$ still applies.
This parameter may grow as our
analysis progresses through the program, selecting approximate couplings for the
sampling instructions as we go. If we have two states $(s_1, s_2)$ that have
probabilities that are bounded by $c$ and we then pair the results from the next
sampling instruction $(v_1, v_2)$ with probabilities bounded by $c'$, the
resulting paired states will have probabilities bounded by $c + c'$.  We can
intuitively think of $c'$ as a \emph{cost} that we need to \emph{pay} in order
to pair up the samples $(v_1, v_2)$ from two sampling instructions.

To make this discussion  more concrete, suppose we sample from the two Laplace
distributions in \cref{fig:ex} (bottom) in the two runs. This figure depicts
the probability density functions (\abr{PDF}) of two different \emph{Laplace
  distributions} over the integers $\mathds{Z}$:%
  \footnote{%
  For visualization, we use the continuous Laplace
  distribution; formally, we assume the discrete Laplace distribution.}
  $\mu_0$ (red; solid) is
centered at $0$ and  $\mu_1$ (blue; dashed) is centered at $1$.
Then, the relation
\[
\Lambda_= = \{(x,x,c) \mid x \in \mathds{Z}\}
\]
is a possible approximate coupling for the two distributions, denoted
$\mu_0\aspace{\Lambda_=}\mu_1$, where the cost $c$ depends on the
width of the Laplace distributions. The coupling $\Lambda_=$ pairs every sample
$x \in \mathds{Z}$ from $\mu_0$ with the same sample $x$ from $\mu_1$. These
samples have different probabilities---since $\mu_0$ and
$\mu_1$ are not equal---so the coupling records that the two probabilities
are within an $\exp(c)$ multiple of one another.\footnote{%
  When $c$ is small, this is roughly a $(1 + c)$ factor.}

In general, two distributions have multiple possible approximate couplings.
Another valid coupling for $\mu_0$ and $\mu_1$ is
\[
  \Lambda_{+1} = \{(x,x+1,0) \mid x \in \mathds{Z}\} .
\]
This coupling pairs each sample $x$ from $\mu_0$ with the
sample $x+1$ from $\mu_1$. Since $\mu_1$ is just $\mu_0$ shifted by $1$, we know
that the probabilities of sampling $x$ from $\mu_0$ and $x+1$ from $\mu_1$ are
in fact \emph{equal}. Hence, each pair has cost $0$ and we incur no cost from
selecting this coupling.

\paragraph{Concrete Example.}
%Let us see how to apply these ideas to prove a special case of differential
%privacy for our running example.
Let us consider the concrete case of \textsf{rnm} with adjacent inputs $q_1 =
[0,1]$ and $q_2 = [1,0]$, and focus on one possible return value $j = 0$. To
prove the probability bound for output $j = 0$, we need to choose couplings for
the sampling instructions such that every pair of linked output states has cost
$c \leq \epsilon$, and if the first output is $0$, then so is the second output.
In other words, the coupled outputs should satisfy the relation $r_1 = 0
\implies r_2 = 0$.
%
% Different choices of intermediate couplings can lead to different costs. We will
% first work through a simple choice that achieves the target relation but with
% cost that is too large. Then, we will refine our choice to show a lower, more
% precise cost.

In the first iteration, our analysis arrives at the sampling instruction $d \sim
\lap_{\eps/2}(q[0])$, representing a sample from the Laplace distribution with
mean $q[0]$ and \emph{scale} $2/\eps$.  Since $q_1[0]$ and $q_2[0]$ are $1$
apart, we can couple the values of $d_1$ and $d_2$ with the first coupling
\[
  \Lambda_= = \{(x,x,\eps/2) \mid x \in \mathds{Z}\} .
\]
By paying a cost $\eps/2$, we can assume that $d_1 = d_2$; effectively,
applying the coupling $\Lambda_=$ as replaces the probabilistic
assignments into a single assignment, where the two processes
non-deterministically choose a coupled pair of values from $\Lambda_=$.
Therefore, the
processes enter the same branch of the conditional and always set $r$ to the
same value---in particular, if $r_1 = 0$, then $r_2 = 0$, as desired.
Similarly, in the
second iteration, we can select the same coupling $\Lambda_=$ and pay $\eps/2$.
This brings our total cost to $\eps/2 + \eps/2 = \eps$ and ensures that both
processes return the same value $r_1 = r_2$---in particular, if $r_1 = 0$ then
$r_2 = 0$, again establishing the requirement for \edp for output $j = 0$.

This assignment of a coupling for each iteration is an example of a
\emph{coupling strategy}: namely, our simple strategy selects $\Lambda_=$ in
every iteration. This strategy also applies for more general inputs, but
possibly with a different total cost. For instance if we consider larger inputs, e.g.,
with $|q| = 3$, we will have to pay $\eps/2$ in three
iterations, and therefore we cannot conclude \edp.  Indeed, this strategy only
establishes $\frac{|q|\eps}{2}$-differential privacy in general, as we pay
$\eps/2$ in each of the $|q|$ loop iterations. Since larger parameters
correspond to looser differential privacy guarantees, this is a weaker property.
However, Report Noisy Max is in fact \edp for arrays of any length. To
prove this stronger guarantee, we need a more sophisticated coupling strategy.

We call coupling strategies that establish \edp \emph{winning couplings
  strategies}.
Intuitively, the verifier plays a game against the two processes:
The processes play by making non-deterministic choices,
while the verifier plays by selecting and paying for approximate couplings
of sampling instructions, aiming to stay under a ``budget'' $\eps$.

\subsection{Synthesizing winning coupling strategies}\label{ssec:details}

% Now, let us see a winning coupling strategy for \textsf{rnm} and
% discuss how our approach can
% synthesize it automatically.

% \paragraph{A winning coupling strategy.}
%
Let us consider a different coupling strategy for \textsf{rnm}.
Informally, we will focus on one possible output $j$ at a time, and we will
only pay non-zero cost in the iteration where $r_1$ may be set to $j$.  By
paying for just a single iteration, our cost will be independent of the number
of iterations $|q|$. Consider the following winning coupling strategy:
\begin{enumerate}
  \item In loop iterations where $i_1 \neq j$, select the following coupling
    (known as the \emph{null coupling}):
    \[
      \Lambda_\varnothing = \{(x_1,x_2,0) \mid x_1 - x_2 = q_1[i_1] - q_2[i_2], x_1,x_2\in \mathds{Z}\}
    \]
    This coupling incurs $0$ privacy cost
    %---it is a more general version of the
    %coupling $\Lambda_{+1}$ we saw before---
    while ensuring that the two samples
    remain the same distance apart as $q_1[i_1]$ and $q_2[i_2]$.
  \item In iteration $i_1 = j$, select the following coupling
    (known as the \emph{shift coupling}):
    \[
      \Lambda_{+1} = \{(x, x + 1, \eps) \mid x \in \mathds{Z}\}
    \]
\end{enumerate}

Let us explain the idea behind this strategy. In all iterations
where $i_1 \neq j$, the processes choose two linked values in
$\Lambda_\varnothing$ that differ by at
most $1$ (since  $|q_1[i_1] - q_2[i_2]| \leq 1$). So the two processes may take different branches of the conditional,
and therefore disagree on the values of $r_1$ and $r_2$.  When $i_1 = j$, we
select the shift coupling $\Lambda_{+1}$ to ensure that $d_2 = d_1 + 1$.
Therefore, if $r_1$ is set to $j$, then $r_2$ is set to $j$;
this is because \rone $i_1=i_2$
and  \rtwo $d_2$ will have to be greater than any of the largest
values ($\emph{best}_2$) encountered by the second execution,
forcing the second execution to enter the \emph{then}
branch and update $r_2$.
Thus, this
coupling strategy ensures that if $r_1 = j$ at the end of the execution, then
$r_2 = j$. Since a privacy cost of $\epsilon$ is only incurred for the
single iteration when $i_1 = j$, each pair of linked outputs has cost $\eps$. This
establishes \edp of \textsf{rnm}.

\paragraph{Constraint-Based Formulation of Coupling Proofs.}
%
% Now, let us see how to synthesize these coupling strategies.
We can view the coupling strategy as assigning a first-order relation
\[
  \strategy(\cvars,z_1,z_2,\distance)
\]
over $(z_1, z_2, \distance)$ to every sampling instruction in the program.  The
vector of variables $\cvars$ contains two copies the program variables, along
with two logical variables representing the output $\ret$ we are analyzing and
the privacy parameter $\epsilon$.

Fixing a particular program state $\vec{c}$, we
interpret $\strategy(\vec{c},z_1,z_2,\distance)$ as a ternary relation over
$\mathds{Z} \times \mathds{Z} \times \mathds{R}$---a coupling of the two
(integer-valued) distributions in the associated sampling instruction, where
$\distance$ records the cost of this coupling. Effectively, the
relation $\strategy$ encodes a function that takes the current state of the two
processes $\cvars$ as an input and returns a coupling.

Finding a winning coupling strategy, then, boils down to finding an appropriate
interpretation of the relation $\strategy$.  To do so, we generate a system of
constraints $\clss$ whose solutions are winning coupling strategies. A solution
of $\clss$ encodes a proof of \edp. For instance, our implementation finds the
following solution of $\strategy$ for \textsf{rnm} (simplified for illustration):
\[
  \bigwedge
    \begin{array}{l}
      i_1 \neq \ret \implies z_1 - z_2 = q_1[i_1] - q_2[i_2] \land \distance = 0\\
      i_1 = \ret \implies z_2 = z_1 + 1 \land \distance = 2 \cdot (\eps/2)
    \end{array}
\]
This is a first-order Boolean formula whose satisfying assignments are elements of
$\strategy$; the two conjuncts encode the null and shift coupling, respectively.

\paragraph{Horn Modulo Couplings.}
More generally, we work with a novel combination of first-order-logic
constraints---in the form of \emph{constrained Horn clauses}---and probabilistic
\emph{coupling constraints}. We call our constraints \emph{Horn modulo couplings}
(\hmc) constraints.

Below, we illustrate the more interesting $\hmc$ constraints generated by our
system for \textsf{rnm}. Suppose that the unknown relation $\inv_\ell$ captures the
\emph{coupled invariant} of the two processes at line number $\ell$.
Just like an
invariant for a sequential program encodes the set of reachable states at a
program location, the coupled invariant describes the set of reachable states of
the two executions (along with privacy cost $\costvar$) assuming a particular
coupling of the sampling instructions.
\begin{align*}
  \cls_1 &: \inv_2(\cvars,\costvar) \land i_1 < |q_1| \land i_2 < |q_2| \longrightarrow \inv_3(\cvars,\omega)\\
  \cls_2 &:   \inv_7(\cvars,\costvar) \longrightarrow \costvar \leq \epsilon
  \land (r_1 = \ret \Rightarrow r_2 = \ret)\\
  \cls_3 &:   \inv_3(\cvars,\costvar) \land \strategy(\cvars,z_1,z_2,\distance)
    \land \costvar' = \costvar+\distance
    \longrightarrow \inv_4(\cvars[d_1\mapsto z_1, d_2 \mapsto z_2],\costvar')
\end{align*}
Informally, a Horn clause is an implication ($\to$)
describing the flow of states
between consecutive lines of the program,
and constraints on permissible states at each program location.
Clause $\cls_1$ encodes the loop entry condition:
if both processes satisfy the loop-entry condition
at line 2, then both states are propagated to line 3.
Clause $\cls_2$
encodes the conditions for \edp at the end of the program (line 7):  the
coupled invariant implies that
incurred privacy cost $\costvar$ is at most $\epsilon$, and if the first process
returns $\ret$, then so does the second process.
% \footnote{%
%   This constraint reflects the asymmetry in the definition
%   of differential privacy. We could alternatively use the simpler $r_1 = r_2$,
%   but this would be overly strong for some proofs.}
Clause $\cls_3$ encodes the effect of selecting the coupling from the strategy
at line $3$---$\costvar$ is incremented by the privacy cost $\distance$ and
$d_1,d_2$ are updated to new values $z_1,z_2$, non-deterministically chosen from
the coupling provided by the strategy.

Clauses 1--3 are standard constrained Horn clauses, but they are not enough: we
need to ensure that $\strategy$ encodes a coupling. So, our constraint system
also features probabilistic \emph{coupling constraints}:
\[
  \cls_4 :   \lap_{\eps/2}(q_1[i_1]) \aspace{\strategy(\cvars,-,-,-)} \lap_{\eps/2}(q_2[i_2])
\]
The constraint says that if we fix the first argument $\cvars$ of $\strategy$ to
be the current state, the resulting ternary relation is a coupling for the two
distributions $\lap_{\eps/2}(q_1[i_1])$ and $\lap_{\eps/2}(q_2[i_2])$.

\paragraph{Solving \hmc Constraints.}
To solve \hmc constraints, we transform the
coupling constraints into standard Horn clauses with unknown expressions
(uninterpreted functions).
% By
% carefully performing such a transformation, we ensure that any solution of the
% unknown expressions results in a coupling strategy.
To sketch the idea, observe that a strategy can be encoded by
a formula of the following form:
%
% \jh{Trying to pin this down: we say ``usually'', but isn't the strategy
%   \emph{always} of the following form? That is, we introduce the incompleteness
%   later when we restrict $\Psi$ to be from a certain family of couplings,
%   right?}
%
\[
  (\varphi_1 \Rightarrow \Psi_1) \land \cdots \land (\varphi_n \Rightarrow \Psi_n)
\]
This describes a \emph{case statement}: if the state of the processes satisfies
$\varphi_i$, then select coupling $\Psi_i$. To find formulas
$\varphi_i$ and $\Psi_i$, we transform solutions of $\strategy$ to
a formula of the above form, parameterized by unknown expressions.
By syntactically restricting $\Psi_i$ to known couplings, we
end up with a system of Horn clauses with unknown expressions in $\varphi_i$ and
$\Psi_i$ whose satisfying assignments always satisfy the coupling constraints.

Finally, solving a system of Horn clauses with unknown expressions is similar to
program synthesis for programs with holes, where we want to find a
completion that results in a correct program; here, we
want to synthesize a coupling strategy proving \edp.
To implement our system, we leverage automated techniques from program synthesis
and program verification.

%!TEX root=paper.tex

\section{Differential Privacy and Variable Approximate Couplings} \label{sec:prelims}
%!TEX root=paper.tex

In this section we
define our program model (\cref{ssec:model}) and $\eps$-differential privacy
(\cref{ssec:dp}). We then present \emph{variable approximate couplings} (\cref{ssec:vac}), which
generalize the approximate couplings of~\citet{BEHSS17,BGGHS16} and form the
theoretical foundation for our coupling strategies proof method.

\subsection{Program Model}\label{ssec:model}

\paragraph{Probability Distributions.}
To model probabilistic computation, we will use probability distributions and
sub-distributions. A function $\mu : \aset \to [0,1]$ defines a
\emph{sub-distribution} over a countable set $\aset$ if $\sum_{\aelem \in \aset}
\mu(\aelem) \leq 1$; when $\sum_{\aelem \in \aset} \mu(\aelem) = 1$, we call
$\mu$ a \emph{distribution}. We will often write $\mu(A)$ for a subset $A
\subseteq \aset$ to mean $\sum_{x \in A} \mu(x)$. We write $\sdist(\aset)$ and
$\dist(\aset)$ for the set of all sub-distributions and distributions over
$\aset$, respectively.

We will use a few standard constructions on sub-distributions.  First, the
\emph{support} of a sub-distribution $\mu$ is defined as  $\supp(\mu) = \{\aelem
\in \aset \mid \mu(\aelem) > 0\}$.  Second, for a sub-distribution on pairs $\mu
\in \sdist(\aset_1 \times \aset_2)$, the first and second \emph{marginals} of
$\mu$, denoted $\pi_1(\mu)$ and $\pi_2(\mu)$, are sub-distributions in
$\sdist(\aset_1)$ and $\sdist(\aset_2)$:
\begin{align*}
  \pi_1(\mu)(\aelem_1) \triangleq \sum_{\aelem_2 \in \aset_2}\mu(\aelem_1,\aelem_2)
  \hspace{1in}
  \pi_2(\mu)(\aelem_2) \triangleq \sum_{\aelem_1 \in \aset_1}\mu(\aelem_1,\aelem_2) .
\end{align*}
We will use $M: W \to \sdist(B)$ to denote a \emph{distribution family}, which
is a function mapping every parameter in some set of \emph{distribution
  parameters} $W$ to a sub-distribution in $\sdist(B)$.

\paragraph{Variables and Expressions.}
We will work with a simple computational model where programs are labeled
graphs. First, we fix a set $\vars$ of \emph{program variables}, the disjoint
union of \emph{input variables} $\varsi$ and \emph{local variables} $\varsl$.
Input variables model inputs to the program and are never modified---they can
also be viewed as logical variables. We assume that a special, real-valued
input variable
$\epsilon$ contains the target privacy level. The set of local variables
includes an \emph{output variable} $\retv \in \varsl$, representing the value
returned by the program. We assume that each variable $v$ has a type $\dom_v$,
e.g., the Booleans $\mathds{B}$, the natural numbers $\mathds{N}$, or the
integers $\mathds{Z}$.

We will consider several kinds of expressions, built out of variables. We will
fix a collection of primitive operations; for instance, arithmetic operations
$+$ and $\cdot$, Boolean operations $=, \land, \lor$,
%constructors for pairs
%$(-, -)$ and lists $[], cons(-, -)$,
etc.
We shall use $\expr$ to denote
expressions over variables $\vars$, $\bexpr$ to denote Boolean expressions over
$\vars$, and $\iexpr$ to denote input expressions over $\varsi$.
We consider a set of
distribution expressions $\dexpr$ modeling primitive, built-in distributions.
For our purposes, distribution expressions will model standard distributions
used in the differential privacy literature; we will detail these when we
introduce differential privacy below. Distribution expressions can be
parameterized by standard expressions.

Finally, we will need three kinds of basic program statements $\stmt$. An
\emph{assignment statement} $v \gets \expr$ stores the value of $\expr$ into the
variable $v$.  A \emph{sampling statement} $v \sim \dexpr$ takes a fresh sample
from the distribution $\dexpr$ and stores it into $v$. Finally, an \emph{assume
  statement} $\assume(\bexpr)$ does nothing if $\bexpr$ is true, otherwise it
filters out the execution (fails to terminate).

\paragraph{Programs.}
A program $\prog = (\locs, \edges)$ is a directed graph where nodes $\locs$
represent program locations, directed edges $\edges \subseteq \locs \times \locs$ connect pairs of locations, and
each edge $e \in \edges$ is labeled by a statement $\stmt_e$.
There is a distinguished \emph{entry location} $\lentry
\in \locs$, representing the first instruction of the program, and a \emph{return
  location} $\lexit \in \locs$. We assume all locations in $\locs$ can reach
$\lexit$ by following edges in $\edges$.
\begin{wrapfigure}[6]{r}{.33\linewidth}
  \vspace{-.1in}
\hspace{0in}{\includegraphics[scale=1.45]{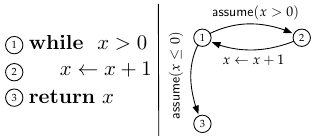}}
\end{wrapfigure}

Our program model is expressive enough to encode the usual conditionals and
loops from imperative languages---e.g., consider the program on the right along
with its graph representation.  All programs we consider model structured
programs, e.g., there is no control-flow non-determinism.

\paragraph{Expression Semantics.}
%To give a semantics to programs, we first define program states and show how to
%interpret variables and expressions in a particular state.
A program \emph{state} $s$
is a map assigning a value to each variable in $\vars$; let $\states$
denote the set of all possible states.
Given variable $v$, we use $s(v)$ to
denote the value of $v$ in state $s$.
Given constant $c$, we use $s[v\mapsto c]$
to denote the state $s$ but with variable $v$
mapped to $c$.
The semantics of an expression $\expr$
is a function $\sem{\expr} : \states \to D$
from a state to an element of some type $\dom$.
For instance, the expression $x + y$ in
state $s$ is interpreted as $\sem{x+y}(s) = \sem{x}(s) + \sem{y}(s)$.
A distribution expression
$\dexpr$ is interpreted as a distribution family
$\sem{\dexpr} : S \to \dist(\mathds{Z})$,
mapping a state in $S$ to a distribution over
integers; for concreteness we will only consider distributions over integers, but
extending to arbitrary discrete distributions is straightforward.
We will often write $s(-)$ to denote $\sem{-}(s)$, for simplicity.

\paragraph{Program Semantics.}
We can define semantics of a program $\prog$ as the aggregate of all its
traces. A \emph{trace} $\trace$ through $\prog$ is a finite sequence of
statements $\stmt_{e_1},\ldots, \stmt_{e_n}$ such that the associated edges
edges $e_1,\ldots,e_n \in \edges$ form a directed path through the graph.
A \emph{maximal trace} $\trace$ has edges corresponding to a path from $\lentry$
to $\lexit$.
We use $\traces(P)$ to denote the set of all maximal traces through $P$.
A trace $\trace$ represents a sequence of instructions to be
executed, so we define $\sem{\trace}: \states \to \sdist(\states)$ as a
distribution family
from states to sub-distributions over states. Assignments and sampling
statements are given the expected semantics, and an $\assume(\bexpr)$ statement
yields the all-zeros sub-distribution if the input state does not satisfy the guard
$\bexpr$. We can now define the semantics of a full program $\sem{P}: \states
\to \dist(\states)$ as
\[
  \sem{P}(s) \triangleq \sum_{\trace \in \traces(P)} \sem{\sigma}(s)
\]
where the sum adds up output sub-distributions from each trace. We will assume
$\prog$ is terminating; in particular, the sub-distributions
$\sem{\trace}(s)$ sum to a proper distribution.

\subsection{Differential Privacy}\label{ssec:dp}
Differential privacy is a quantitative, statistical notion of data privacy
proposed by \citet{DMNS06}; we reformulate their definition in our
program model.
We will fix an \emph{adjacency relation} on input states $\adj
\subseteq \states \times \states$ modeling which inputs should lead to similar
outputs.
%---for the original motivation of differential privacy, $\adj$ could
%relate two states where the input variable contains two databases that differ in
%the data of a single individual, and where all other variables are the same in
%both states.
%
Throughout, we implicitly assume \rone for all $(s_1,s_2)\in \adj$,
$s_1(\eps) = s_2(\eps) > 0$, and \rtwo for all $(s_1,s_2) \in \adj$
and every $c \in \mathds{R}^{>0}$, $(s_1[\eps \mapsto c], s_2[\eps \mapsto c])\in \adj$.
In other words, in adjacent states $\epsilon$ may take any positive value, as
long as it is equal in both states.

% \retv{I used $j$ instead of set $D$ -- i find the link with couplings to be more
%   immediate with this less general definition.}
\begin{definition}[\citet{DMNS06}] \label{def:dp}
  A program $P$ is $\epsilon$-\emph{differentially private} with respect to an
  adjacency relation $\adj$ iff for every $(s_1,s_2) \in \adj$ and output $j \in
  \dom_{\retv}$,
  \[
    \mu_1(\{s \mid s(v_r) = j\})
    \leq \exp(c) \cdot \mu_2(\{s \mid s(v_r) = j\})
  \]
  where $\mu_1 = \prog(s_1)$, $\mu_2 = \prog(s_2)$, and $s_1(\eps) = s_2(\eps) = c$.
  %We will typically leave
  %the adjacency relation implicit.
\end{definition}

\citet{DMNS06} originally defined differential privacy in terms of \emph{sets}
of outputs rather than single outputs. \cref{def:dp}
is equivalent to their notion of $(\epsilon, 0)$-differential privacy and is more
convenient for our purposes, as we will often prove privacy by
focusing on one output at at time.

We will consider two primitive distributions that are commonly used in
differentially private algorithms: the (discrete) \emph{Laplace} and
\emph{exponential} distributions. Both distributions
are parameterized by two numbers, describing the spread of the distribution and
its center.
%We
%model the two distributions as distribution expressions in our language.

% \retv{i changed $\eps$ to $c$ in the following definitions for generality}

\begin{definition}
  The \emph{Laplace distribution} family is a function with two parameters, $y
  \in \mathds{R}^{>0}$ and $z \in \mathds{Z}$.  We use $\lap_y(z)$ to denote the
  distribution over $\mathds{Z}$ where $\lap_y(z)(\nu) \propto \exp(- |\nu -
  z| \cdot y)$, for all $\nu \in \mathds{Z}$.  We call $1/y$ the \emph{scale} of
  the distribution and $z$ the \emph{mean} of the distribution.
\end{definition}

\begin{definition}
  The \emph{exponential distribution} family is a function with two parameters,
  $y \in \mathds{R}^{>0}$ and $z \in \mathds{Z}$.  We use $\olap_{y}(z)$ to
  denote the distribution over $\mathds{Z}$ where $\olap_y(z)(\nu) = 0$ for all
  $\nu < z$, and $\olap_{y}(z)(\nu) \propto \exp(- (\nu - z) \cdot y)$ for all
  $\nu \geq z$.
\end{definition}

We will use distribution expressions
$\lap_{\iexpr}(\expr)$  and $\olap_{\iexpr}(\expr)$
to represent the respective Laplace/exponential distribution
interpreted in a given state.
We assume that every pair of adjacent
states $(s_1,s_2) \in \adj$ satisfies $s_1(\iexpr) = s_2(\iexpr) > 0$
for every  $\iexpr$ appearing in distribution expressions.

\subsection{Variable Approximate Couplings}\label{ssec:vac}

We are now ready to define the concept of a
\emph{variable approximate coupling}, which we will
use to establish \edp of a program.
A variable approximate coupling
of two distributions $\mu_1$ and $\mu_2$
can be viewed as two distributions on pairs, each modeling one of the original
distributions.
Formally:
% In
% this way, we may compare the probability of an event $E_1$ in the first output
% distribution to the probability of an event an event $E_2$ in the second output
% distribution by finding an approximate coupling pairing samples leading to $E_1$
% and samples leading to $E_2$.

%More formally, an approximate coupling for two distributions $\mu_1$ and $\mu_2$
%can be viewed as two distributions on pairs, each modeling one of the original
%distributions.

\begin{definition}[Variable approximate couplings]
  Let $\mu_1 \in \sdist(\aset_1)$ and $\mu_2 \in \sdist(\aset_2)$. Let $\Lambda
  \subseteq B_1 \times B_2 \times \mathds{R}^{\geq 0}$ be a relation. We write
  $dom(\Lambda) \subseteq B_1 \times B_2$ for the projection of $\Lambda$ to the
  first two components.  We say that $\Lambda$ is a \emph{variable approximate
    coupling} of $\mu_1$ and $\mu_2$---denoted $\mu_1 \aspace{\Lambda}
  \mu_2$---if there exist two \emph{witness} sub-distributions $\mu_L \in \sdist(\aset_1 \times
  \aset_2)$ and $\mu_R \in \sdist(\aset_1 \times \aset_2)$ such that:
  \begin{enumerate}
    \item $\pi_1(\mu_L) = \mu_1$ and $\pi_2(\mu_R) \leq \mu_2$,
    \item $\supp(\mu_L), \supp(\mu_R) \subseteq dom(\Lambda)$, and
    \item for all $(\aelem_1,\aelem_2,c) \in \Lambda$, we have $\mu_L(\aelem_1,
      \aelem_2) \leq \exp(c) \cdot \mu_R(\aelem_1, \aelem_2)$.
  \end{enumerate}
  % We will write  $\Lambda^c$ for $c \in \mathds{R}^{\geq 0}$ when $c' \leq c$
  % for all $(\aelem_1,\aelem_2,c') \in \Lambda$.
  We will call these the \emph{marginal} conditions, the \emph{support}
  conditions, and the \emph{distance} conditions respectively. We will often
  abbreviate ``variable approximate coupling'' as simply ``coupling''.
\end{definition}

To give some intuition, the first point states that $\mu_L$ models the first
distribution and $\mu_R$ is a lower bound on the second distribution---the
inequality $\pi_2(\mu_R) \leq \mu_2$ means that $\pi_2(\mu_R)(b_2) \leq
\mu_2(b_2)$ for every $b_2 \in B_2$. Informally, since the definition of
differential privacy (\cref{def:dp}) is asymmetric, it is enough to
show that the first distribution (which we will model by $\mu_L$) is less than a lower
bound of the second distribution (which we will model by $\mu_R$). The second
point states that every pair of elements with non-zero probability in $\mu_L$ or
$\mu_R$ must have at least one cost. Finally, the
third point states that $\mu_L$ is approximately upper bounded by $\mu_R$; the
approximation is determined by a cost $c$ at each pair.

Our definition is a richer version of $\star$-lifting, an approximate coupling
recently proposed by \citet{BEHSS17} for verifying differential privacy. The
main difference is our approximation level $c$ may vary over the pairs $(b_1,
b_2)$, hence we call our approximate coupling a \emph{variable} approximate
coupling. When all approximation levels are the same, our definition recovers
the existing definition of $\star$-lifting.\footnote{%
  More precisely, $(\epsilon, 0)$ $\star$-lifting.}
Variable approximate couplings can give more precise bounds when comparing
events in the first distribution to events in the second distribution.

Now, \edp holds if we can find a coupling of the output distributions from
adjacent inputs.
\begin{lemma}[\edp and couplings]\label{lemma:dpcoupling}
  A program $P$ is \edp
  with respect to adjacency relation $\adj$ if for every $(s_1,s_2) \in \adj$ and
  $j \in \dom_\retv$, there is a coupling $P(s_1)
  \aspace{\Lambda_j} P(s_2)$ with
  \[
    \Lambda_j
    = \{(s_1',s_2',c) \mid s_1'(\retv) = j \Rightarrow s_2'(\retv) = j\},
  \]
  where $s_1(\eps) = s_2(\eps) = c$.
  %\footnote{
  %\iftoggle{long}
  %{All missing proofs are in the Appendix.}
  %{All missing proofs are in a long version in the supplementary materials.}
  %}
\end{lemma}

We close this section with some examples of specific couplings for the Laplace
and exponential distributions proposed by \citet{BGGHS16}.

\begin{definition}[Shift coupling] \label{def:shift}
  Let $y \in \mathds{R}^{> 0}$, $z_1, z_2 \in \mathds{Z}$, and $k \in
  \mathds{Z}$, and define the relations
  \[
    \Lambda_{+k} \triangleq \{ (n_1, n_2, |k + z_1 - z_2| \cdot y) \mid n_1 + k = n_2 \}
    \qquad \text{and} \qquad
    \Lambda'_{+k} \triangleq \{ (n_1, n_2, (k + z_1 - z_2) \cdot y) \mid n_1 + k = n_2 \}
  \]
  Then we have the following \emph{shift} coupling for the Laplace distribution:
  \[
    \lap_y(z_1)
    \aspace{\Lambda_{+k}}
    \lap_y(z_2) .
  \]
  If $k + z_1 - z_2 \geq 0$, we have the following shift coupling for the
  exponential distribution:
  \[
    \olap_y(z_1)
    \aspace{\Lambda'_{+k}}
    \olap_y(z_2) .
  \]
  Intuitively, these couplings relate each sample $n_1$ from the
  first distribution with the sample $n_2 = n_1 + k$ from the second
  distribution. The difference in probabilities for these two samples depends on
  the shift $k$ and the difference between the means $z_1, z_2$.
\end{definition}

If we set $k = z_2 - z_1$ above, the coupling cost is $0$ and we have the
following useful special case.

\begin{definition}[Null coupling]\label{def:null}
  Let $y \in \mathds{R}^{> 0}$, $z_1, z_2 \in \mathds{Z}$, and
  define the relation:
  \[
    \Lambda_\varnothing \triangleq \{ (n_1, n_2, 0) \mid n_1 - z_1 = n_2 - z_2 \} .
  \]
  Then we have the following \emph{null} couplings:
  \[
    \lap_y(z_1)
    \aspace{\Lambda_\varnothing}
    \lap_y(z_2)
    \quad\text{and}\quad
    \olap_y(z_1)
    \aspace{\Lambda_\varnothing}
    \olap_y(z_2) .
  \]
  Intuitively, these couplings relate pairs of samples $n_1, n_2$
  that are at the same distance from their respective means $z_1, z_2$---the
  approximation level is $0$ since linked samples have the same probabilities
  under their respective distributions.
\end{definition}

\section{Coupling Strategies}\label{sec:strategies}

In this section, we introduce \emph{coupling strategies}, our proof technique
for establishing \edp.
%We begin by formalizing coupling
%strategies and \emph{coupled postconditions}.  We then define \emph{winning
%  coupling strategies}, which imply $\epsilon$-differential privacy.

\subsection{Formalizing Coupling Strategies}
% The following definition formalizes coupling strategies.
Roughly speaking, a coupling strategy picks a coupling for each sampling statement.

\begin{definition}[Coupling strategies]\label{def:cpost}
Let $\stmts$ be the set of all sampling statements in a program $\prog$; recall
that our primitive distributions are over the integers $\mathds{Z}$. A
\emph{coupling strategy} $\strp$ for $\prog$ is a map from $\stmts \times
\states \times \states$ to couplings in $2^{\mathds{Z} \times
\mathds{Z} \times \mathds{R}}$, such that for a statement $\stmt = v \sim \dexpr$
and states $(s_1,s_2)$, the relation $\strp(\stmt,s_1,s_2) \subseteq \mathds{Z} \times
\mathds{Z} \times \mathds{R}$
forms a coupling  $s_1(\dexpr) \aspace{\strp(\stmt,s_1,s_2)} s_2(\dexpr) .$
\end{definition}

\begin{example}
  Recall our simple coupling strategy for Report Noisy Max in
  \cref{ssec:coupling}.  For the statement $\stmt = v \sim \lap_{\eps/2}(q[i])$,
  and every pair of states $(s_1,s_2)$ where $s_1(\eps) = s_2(\eps) = c > 0$, the
  strategy  $\strp$ returned the  coupling $\strp(\stmt,s_1,s_2) =
  \{(x,x,c) \mid x \in \mathds{Z}\}$.
\end{example}

To describe the effect of a coupling strategy on two executions from neighboring
inputs, we define a \emph{coupled postcondition} operation. This operation
propagates a set of pairs of coupled states while tracking the privacy cost
from couplings selected along the execution path.

\begin{definition}[Coupled postcondition]
  Let $\strp$ be a coupling strategy for
  $\prog$. We define the \emph{coupled postcondition} as a function
  $\cpost_\strp$, mapping $Q \subseteq \states \times \states \times \mathds{R}$
  and a statement to a subset of $\states \times \states \times \mathds{R}$:
  \begin{align*}
    \cpost_\strp(Q, v \gets \expr) \triangleq{}&
    \{ (s_1[v \mapsto s_1(\expr)], s_2[v \mapsto s_2(\expr)], c) \mid (s_1, s_2, c) \in Q \} \\
    \cpost_\strp(Q, \assume(\bexpr)) \triangleq{}&
    \{ (s_1, s_2, c) \mid (s_1, s_2, c) \in Q \land s_1(\bexpr) \land s_2(\bexpr) \} \\
    \cpost_\strp(Q, v \sim \dexpr) \triangleq{}&
    \{ (s_1[v \mapsto a_1], s_2[v \mapsto a_2], c + c') \\
    &\mid (s_1, s_2, c) \in Q \land (a_1, a_2, c') \in
    \strp(v \sim \dexpr, s_1, s_2) \} .
  \end{align*}
  This operation can be lifted to operate on sequences of statements $\sigma$ in
  the standard way: when the trace $\trace = \stmt_1\trace'$ begins
  with $\stmt_1$, we define $\cpost_\strp(Q, \trace) \triangleq
  \cpost_\strp(\cpost_\strp(Q,\stmt_1),\trace')$;
  when $\sigma$ is the empty trace, $\cpost_\strp(Q, \trace) \triangleq Q$.
\end{definition}

Intuitively, on assignment statements $\cpost$
updates every pair of states $(s_1,s_2)$ using the semantics of
assignment;
on assume statements, $\cpost$
only propagates pairs of states satisfying $\bexpr$;
on sampling statements,
$\cpost$ updates every pair of states $(s_1,s_2)$
 by assigning the variables $(v_1,v_2)$
 with every possible pair from the coupling
chosen by $\tau$, updating
the incurred privacy costs.

\begin{example}\label{ex:cpost}
  Consider a simple program:
  \begin{algorithmic}
    \State $x \gets x + 10$
    \State $x \sim \lap_\eps(x)$
  \end{algorithmic}
  Let the adjacency relation
  $\Delta = \{(s_1,s_2) \mid s_1(x) = s_2(x) + 1 \text { and }s_1(\eps) = s_2(\eps) = c\}$
  and let $\strp$ be a coupling strategy.
  Since there is only a single variable $x$ and $\eps$ is fixed to $c$, we will represent the states of the two
  processes by $x_1$ and $x_2$, respectively.
  Then, the initial
  set of coupled states $Q_0$ is the set $\{(x_1,x_2,0)
  \mid x_1 = x_2 + 1\}$.  First, we compute
  $Q_1 = \cpost_\strp(Q_0, x \gets x + 10).$
  This results in a set $Q_1 = Q_0$, since adding 10 to
  both variables does not change the fact that $x_1 = x_2 + 1$.
  Next, we compute
  $Q_2 = \cpost_\strp(Q_1, x \sim \dexpr).$
  Suppose that the coupling strategy $\strp$ maps every pair of states $(s_1,
  s_2) \in Q_1$ to the coupling $\Lambda_= = \{(k,k,c) \mid k \in \mathds{Z}\}$.
  Then,
  $
  Q_2 = \{(x_1,x_2,c) \mid x_1 = x_2\}.
  $
\end{example}

While a coupling strategy restricts the pairs of executions we must consider,
the executions may still behave quite differently---for instance, they may
take different paths at conditional statements or run for different numbers of
loop iterations, etc. To further simplify the verification task we will focus
on \emph{synchronizing} strategies only,
% Intuitively, a
% synchronizing strategy stipulates that both processes follow the same control-flow
% paths through the program, ensuring that we always relate the same instructions in
% both executions.
%
which ensure that the guard in every $\assume$ instruction takes the same value
on coupled pairs, so both coupled processes follow the same control-flow path
through the program.

\begin{definition}[Synchronizing coupling strategies]
  Given a trace $\trace$, let $\trace_i$ denote the prefix
  $\stmt_1,\ldots,\stmt_i$ of $\trace$ and let $\trace_0$ denote the empty trace
  that simply returns the input state.
  A coupling strategy $\strp$ is \emph{synchronizing} for $\trace$ and
  $(s_1,s_2) \in \adj$ iff for every $\stmt_i = \assume(\bexpr)$ and $(s_1',
  s_2', -) \in \cpost_\strp(\{ (s_1, s_2, 0) \},\sigma_{i-1})$, we have
  $s_1'(\bexpr) = s_2'(\bexpr)$.
\end{definition}

Requiring synchronizing coupling strategies is a certainly a restriction, but
not a serious one for proving differential privacy. By treating conditional
statements as monolithic instructions, the only potential branching we need to
consider comes from loops. Differentially private algorithms use straightforward
iteration; all examples we are aware of can be encoded with for-loops and
analyzed synchronously. (\Cref{sec:evaluation} provides more details about how
we handle branching.)

\subsection{Proving Differential Privacy via Coupling Strategies}

Our main theoretical result shows that synchronizing coupling
strategies encode couplings.

\begin{lemma}[From strategies to couplings] \label{lemma:strat-couple}
Suppose $\strp$ is synchronizing for $\trace$ and $(s_1,s_2) \in \adj$.  Let
$\Psi = \cpost_\strp( \{ (s_1, s_2, 0) \}, \sigma )$, and let $f : \states
\times \states \to \mathds{R}$ be such that $c \leq f(s_1', s_2')$ for all
$(s_1', s_2', c) \in \Psi$. Then we have a coupling
$
  \sigma(s_1) \aspace{\Psi_f} \sigma(s_2),
$
where $\Psi_f \triangleq \{ (s_1', s_2', f(s_1', s_2')) \mid (s_1', s_2', -) \in \Psi \}$.
\end{lemma}
\begin{proof}
  (\emph{Sketch}) By induction on the length of the trace $\sigma$.%; 
  %we defer details to
  %the \iftoggle{long}{Appendix}{supplementary materials}.
\end{proof}

There may be multiple sampling choices that yield the same coupled outputs
$(s_1', s_2')$; the coupling must assign a cost larger than all associated
costs in the coupled post-condition.
By \cref{lemma:dpcoupling}, if we can find a coupling for each possible output
$j \in \dom_\retv$ such that $s_1(\retv) = j \implies s_2(\retv) = j$ with cost
at most $\epsilon$, then we establish $\epsilon$-differential privacy.
We formalize this family of coupling strategies as a \emph{winning coupling
  strategy}.

\begin{theorem}[Winning coupling strategies] \label{thm:winning}
Fix program $\prog$ and adjacency relation $\adj$.
Suppose that we have a family of coupling strategies $\{ \strp_j
\}_{j \in \dom_\retv}$ such that for every trace $\trace \in \traces(P)$ and $j
\in \dom_\retv$, $\strp_j$ is synchronizing for $\trace$ and
\[
\cpost_{\strp_j} (\adj \times \{0\},\trace) \subseteq
\{(s_1',s_2',c) \mid c \leq s_1(\eps) \land s_1'(\retv) = j \Rightarrow s_2'(\retv) = j\} .
\]
Then $\prog$ is $\eps$-differentially private.
(We could  change $s_1(\eps)$ to $s_2(\eps)$ above,
since $s_1(\eps) = s_2(\eps)$ in $\Delta$.)
\end{theorem}
%
% Intuitively, for every
% trace $\trace$, output $j \in \dom_\retv$, and pair of initial states $(s_1,s_2)
% \in\adj$ with initial cost $0$, the strategy $\strp_j$ constructs a coupling
% over the output distributions $\trace(s_1)$ and $\trace(s_2)$ such that for any
% pair of linked output states $(s_1', s_2')$:
% \rone the cost is at most $s_1(\eps)$, and
% \rtwo if  $s_1'(\retv) = j$, then $s_2'(\retv) = j$.
%
\begin{proof}
  Let $(s_1, s_2) \in \adj$ be two input states
  such that $s_1(\eps) = s_2(\eps) = c$. For every $j \in \dom_{\retv}$
  and any sequence $\trace \in \traces(\prog)$, \cref{lemma:strat-couple}
  gives a coupling of the output distributions:
  $
    \sigma(s_1)
    \aspace{ \Lambda }
    \sigma(s_2)
  $
  where $s_1'(\retv) = j \implies s_2'(\retv) = j$ and $c' \leq c$ for
  every $(s_1', s_2', c') \in \Lambda$. By a similar argument as in
  \cref{lemma:dpcoupling}, we have
  \[
    \sigma(s_1)(\{ s_1' \mid s_1'(\retv) = j \})
    \leq \exp(c) \cdot \sigma(s_2)(\{ s_2' \mid s_2'(\retv) = j \}) .
  \]
  Since $\sem{\prog} = \sum_{\sigma \in \Sigma(\prog)} \sem{\sigma}$, summing
  over all traces $\sigma \in \Sigma(\prog)$
  shows $\eps$-differential privacy:
  \[
    \prog(s_1)(\{ s_1' \mid s_1'(\retv) = j \})
    \leq \exp(c) \cdot \prog(s_2)(\{ s_2' \mid s_2'(\retv) = j \}) .
    \qedhere
  \]
\end{proof}

\section{Horn Clauses Modulo Couplings}\label{sec:hcc}

So far, we have seen how to prove differential privacy by finding a winning
coupling strategy. In this section we present a constraint-based
formulation of winning coupling strategies, paving the way to
automation.

% We begin by describing a novel extension of \emph{constrained Horn clauses} with
% \emph{coupling constraints}. Then, we show how to generate these constraints for
% a given program.

\subsection{Enriching Horn clauses with Coupling Constraints}

Constrained Horn clauses
are a standard tool for logically encoding
verification proof rules~\citep{bjorner2015horn,grebenshchikov2012}.
For example, a set of  Horn clauses can describe a loop invariant, a
rely-guarantee proof, an Owicki-Gries proof, a ranking function, etc. Since our proofs---winning
coupling strategies---describe probabilistic couplings, our Horn clauses will
involve a new kind of probabilistic constraint.

\paragraph{Horn Clauses.}
We assume \emph{formulas} are interpreted in some first-order
theory (e.g., linear integer arithmetic) with a set $\rels$ of uninterpreted
\emph{relation} symbols. A \emph{constrained Horn clause} $\cls$, or Horn clause for
short, is a first-order logic formula of the form
\[
\rel_1(\vec{v}_1) \land \rel_2(\vec{v}_2) \land \ldots \land \rel_{n-1}(\vec{v}_{n-1})
\land \varphi \longrightarrow \head_\cls
\]
where:
\begin{itemize}
  \item each relation $\rel_i \in \rels$ is of arity  equal to the length of the
    vector of variables $\vec{v}_i$;
  \item $\varphi$ is an interpreted formula over the first-order theory (e.g., $x > 0$);
  \item the left-hand side of the implication ($\longrightarrow$) is called the
    \emph{body} of $\cls$; and
  \item $H_\cls$, the \emph{head} of $\cls$, is either a relation application
    $\rel_n(\vec{v}_n)$ or an interpreted formula $\varphi'$.
\end{itemize}
We will allow interpreted formulas to contain disjunctions.
For clarity of presentation, we will use $\Rightarrow$ to denote implications in an interpreted formula, and $\to$ to denote the implication in a Horn clause.  All free variables
are assumed to be universally quantified, e.g., $x+y > 0
\longrightarrow \rel(x,y)$ means $\forall x,y \ldotp x + y > 0
\longrightarrow \rel(x,y)$.
%
%\jh{In the interpretation, yes? So with $\Rightarrow$?}
%
For conciseness we often write
$\rel(\cvars,x_1,\ldots,x_m)$ to denote the relation application
$\rel(v_1,\ldots,v_n,x_1,\ldots,x_m)$, where $\cvars$ is the vector of variables
$v_1,\ldots,v_n$.

\paragraph{Semantics.}
We will write $\clss$ for a set of clauses $\{\cls_1,\ldots,\cls_n\}$.
Let $G$ be a graph over relation symbols such that
there is an edge $(r_1,r_2)$ iff $r_1$ appears in the body of
some clause $\cls_i$ and $r_2$ appears in its head.
We say that $\clss$ is \emph{recursive} iff $G$ has a cycle.
The set $\clss$ is \emph{satisfiable} if
there exists an interpretation $\rho$ of relation symbols as relations in the
theory such that every clause
$\cls \in \clss$ is valid.
We say that $\interp$ satisfies $\clss$ (denoted $\interp\models\clss$)
iff for all
$\cls \in \clss$, $\interp \cls$ is valid (i.e., equivalent to $\emph{true}$),
where $\interp \cls$ is $\cls$ with every relation application $\rel(\vec{v})$
replaced by $\interp(\rel(\vec{v}))$.

The next example uses Horn clauses to prove a simple Hoare triple for a
standard, deterministic program.  For a detailed exposition of Horn clauses in
verification, we refer to the excellent survey by~\citet{bjorner2015horn}.

\begin{figure}[t!]
  \centering
  \begin{minipage}{.5\textwidth}
\includegraphics[scale=1.4]{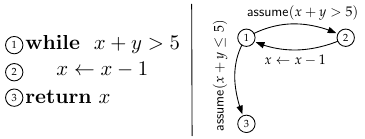}
\end{minipage}
\begin{minipage}{.49\textwidth}
\smaller
\begin{align*}
  \cls_1 &: x = y  \longrightarrow \rel_1(x,y) & \text{precondition}\\
  \cls_2 &: \rel_1(x,y) \land x + y > 5 \longrightarrow \rel_2(x,y)& \text{loop entry}\\
  \cls_3 &: \rel_2(x,y) \land x'=x-1 \longrightarrow \rel_1(x',y) &\text{loop body}\\
  \cls_4 &: \rel_1(x,y) \land x + y \leq 5 \longrightarrow \rel_3(x,y) &\text{loop exit}\\
  \cls_5 &: \rel_3(x,y) \longrightarrow x \leq y &\text{postcondition}
\end{align*}
\end{minipage}
\caption{Example of using Horn clauses for proving a Hoare triple}\label{fig:hoare}
\end{figure}

\begin{example}
Consider the simple program $\prog$ and its graph representation
in \cref{fig:hoare}.
Suppose we want to prove the Hoare triple
$\{x = y\} ~\prog~ \{x \leq y\}$.
  We generate the   Horn clauses
  shown in \cref{fig:hoare},
  where relations $\rel_i$ denote the Hoare-style annotation
  at location $i$ of $\prog$.
  Each clause encodes either the pre/postcondition
  or the semantics of one of the program statements.

  Notice that the constraint system is recursive.
  Intuitively, $\rel_i$ captures all reachable states
  at location $i$; relation $\rel_1$
  encodes an inductive loop invariant that holds at the loop head.
  We can give a straightforward reading of each clause. For instance, $\cls_3$ states that
  if $(x,y)$ is reachable at location $2$, then
  $(x-1,y)$ must be reachable at location $1$;
  $\cls_5$ stipulates that all states
  reachable at location 5 must be such that $x \leq y$,
  i.e., satisfy the postcondition.

  One possible satisfying assignment $\interp$ to these constraints
  is $\interp(\rel_i) = \{(a,b) \mid a \leq b\}$ for $i \in [1,3]$.
  Applying $\interp$ to
  $\rel_1(x,y)$ results in $x \leq y$, the loop invariant.
  We can interpret $\cls_1$ and $\cls_3$ under $\interp$:
  \begin{align*}
    \interp\cls_1 : x = y  \longrightarrow x \leq y &&
    \interp\cls_3 : x \leq y \land x'=x-1 \longrightarrow x' \leq y
  \end{align*}
  Observe that all $\interp\cls_i$ are valid,
  establishing validity of the original Hoare triple.
\end{example}

\paragraph{Horn Clauses Modulo Coupling Constraints.}
We now enrich Horn clauses with a new form of constraint that describes
couplings between pairs of distributions; we call the resulting
constraints \emph{Horn modulo couplings} (\hmc).

\begin{definition}[Coupling constraints]
  Suppose we have a relation $\rel$ of arity $n + 3$ for some $n \geq 0$, and
  suppose we have two distribution families $M_1$ and $M_2$.
  A coupling constraint is of the form:
  \[
    M(\vec{v}_1) \aspace{\rel(\vec{v}_3,-,-,-)} M(\vec{v}_2)
  \]
  An interpretation $\interp$ satisfies such a constraint iff for every
  consistent assignment $\vec{c}_1, \vec{c}_2, \vec{c}_3$ to the variables
  $\vec{v}_1$, $\vec{v}_2$, $\vec{v}_3$, we have
  $
    M(\vec{c}_1)
    \aspace{\interp(\rel(\vec{c}_3,-,-,-))}
    M(\vec{c}_2),
  $
  where $\interp(\rel(\vec{c}_3,-,-,-))$ is defined
  as $\{(k_1,k_2,\theta) \mid (\vec{c}_3,k_1,k_2,\theta) \in \interp(r)\}$.
  In other words, $\interp(\rel(\vec{c}_3,-,-,-))$
  is the ternary relation resulting from fixing the
  first $n$ components of $\interp(\rel)$ to $\vec{c}_3$.
  Note that the vectors $\{\vec{v}_i\}_i$
  may share variables.
  % and we assume that the distribution families
  % $M_1, M_2$ have parameters that are of the same types
  % as the logical variables $\vec{v}_1$ and $\vec{v}_2$, respectively.
\end{definition}

\begin{example}

Consider the simple coupling constraint
$
   \lap_c(x_1) \aspace{\rel(x_1,x_2,-,-,-)} \lap_c(x_2),
$
where $c$ is a positive constant.  A simple satisfying assignment is the null
coupling from \cref{def:null}:
\[
\interp(r) = \{(a_1,a_2,b_1,b_2,0) \mid b_1 - b_2 = a_1 - a_2\}
\]
Effectively, $\interp(r)$ says: \emph{for any values $(a_1,a_2)$ of $(x_1,x_2)$,
  the relation  $\Lambda_\varnothing = \{(b_1,b_2,0) \mid b_1-b_2 = a_1-a_2\}$ is a coupling
  such that $\lap_c(a_1) \aspace{\Lambda_\varnothing} \lap_c(a_2)$}.
\end{example}

\begin{example}

For a simpler example, consider the constraint
$
   \lap_c(x) \aspace{\rel(-,-,-)} \lap_c(x),
$
where $c$ is a positive constant.
A  satisfying assignment to this constraint
is:
$
\interp(r) = \{(b_1,b_2,0) \mid b_1=b_2\},$
since the two distributions are the same
for any value of variable $x$.
Another possible satisfying assignment
is $\interp(r) = \{(b_1,b_2,2c) \mid b_1+2=b_2\}$,
by the shift coupling (\cref{def:shift}).
\end{example}

\begin{figure}[t!]
  \centering
\smaller
\begin{prooftree}
~
\justifies
 \enc{\adj}\land \costvar =0 \to \inv_{\lentry}(\cvars,0)  \in \clss
\using \cinit
\end{prooftree}
\hspace{.5in}
\begin{prooftree}
\justifies
\inv_{\lexit}(\cvars,\costvar)
\to
\costvar \leq \eps \land (\retva = \ret \Rightarrow \retvb = \ret)
\in \clss
\using \cdp
\end{prooftree}

\vspace{.3in}

\begin{prooftree}
e = (i,j) \in \locs
\quad \quad
\stmt_e = v \gets \expr
\justifies
\inv_i(\cvars,\costvar) \land v_1' = \enc{\expr\rel_1}
\land v_2' = \enc{\expr\rel_2}
\to \inv_j(\cvars[v'],\costvar)  \in \clss
\using \cassign
\end{prooftree}

\vspace{.3in}
\begin{prooftree}
e = (i,j) \in \locs
\quad \quad
\stmt_e = \assume(\bexpr)
\justifies
\inv_i(\cvars,\costvar) \land \enc{\bexpr\rel_1} \land \enc{\bexpr\rel_2}
\to \inv_j(\cvars,\costvar)  \in \clss
\using \casmconst
\end{prooftree}
\hspace{.3in}
\begin{prooftree}
e = (i,j) \in \locs
\quad \quad
\stmt_e = \assume(\bexpr)
\justifies
\inv_i(\cvars,\costvar)
\to \enc{\bexpr\rel_1} \equiv \enc{\bexpr\rel_2}  \in \clss
\using \casmsync
\end{prooftree}

\vspace{.3in}

\begin{prooftree}
e = (i,j) \in \locs
\quad \quad
\stmt_e = v \sim \dexpr
\justifies
\inv_i(\cvars,\costvar) \land
\strategy_e(\cvars,v_1',v_2',\distance)
\to \inv_j(\cvars[v'],\costvar + \distance)  \in \clss
\using \cstrategy
\end{prooftree}
\hspace{.15in}
\begin{prooftree}
e = (i,j) \in \locs
\quad \quad
\stmt_e = v \sim \dexpr
\justifies
\enc{\dexpr_1} \ \aspace{\strategy_e(\cvars,-,-,-)}\
\enc{\dexpr_2}
 \in \clss
\using \ccouple
\end{prooftree}

\vspace{.3in}

We use $\cvars[v']$ to denote the vector $\cvars$ but with
variables $v_1$ and $v_2$ replaced by their primed versions,
$v_1'$ and $v_2'$.

\caption{Generating Horn modulo coupling constraints}
\label{fig:hmc}
\end{figure}

\subsection{Generating Horn clauses from Programs}
For a program $\prog$, we generate clauses $\clss$ specifying
a winning coupling strategy $\{ \strp_j \}_{j \in
  \dom_\retv}$; as we saw in \cref{thm:winning}, a winning coupling
strategy implies $\epsilon$-differential privacy.
A winning coupling strategy ensures that the coupled postcondition
$\cpost_{\strp_j}(\adj \times \{0\}, \trace)$ satisfies certain conditions for
every trace $\trace \in \traces(P)$. We thus generate a set of constraints
whose solutions capture all sets $\cpost_{\strp_j}(\adj \times \{0\}, \trace)$.
Since the set $\traces(P)$ is potentially infinite (due to loops), we
generate a recursive system of constraints.
Rather than find a coupling
strategy separately for each possible output value $j$, we will parameterize our
constraints by a logical variable $\ret$ representing a possible output.

% \aws{i think the target cost variable is incorrectly introduced;
% we need to fix an epsilon for every execution; the strategy needs
% to know what epsilon is.}
% \jh{The target cost matters for figuring out whether a strategy is winning or
%   not (so we have the constraint in dpriv). But why does the strategy need to
%   know what the target cost is?}

\paragraph{Invariant and Strategy Relations.}
Our generated constraints $\clss$ mention two unknown relations.
\begin{itemize}
  \item $\inv_i(\cvars,\costvar)$ encodes the \emph{coupled invariant} at
    location $i \in \locs$ in the program.  This relation captures the set
    of coupled states $\cpost_{\strp_j}(\adj \times \{0\}, \trace)$ for traces
    $\trace$ that begin at location $\lentry$ and end at location $i$. The
    first-order variables $\cvars$ model two copies of program variables for the two
    executions of $\prog$, tagged with subscript $1$ or $2$, respectively, along
    with the logical variable $\ret$ representing a possible return value in
    $\dom_\retv$. The variable $\costvar$ models the accumulated
    cost for the particular coupled states and program location.
  \item $\strategy_e(\cvars,v_1',v_2',\distance)$ encodes the coupling strategy
    for sampling statement $\stmt_e$. If the values of $\cvars$ model two
    program states, $\strategy_e$ encodes a coupling between the distributions
    in $\stmt_e$.
\end{itemize}

\paragraph{Constraint Generation.}
Given a program $\prog$ and adjacency relation $\adj$, the rules in
\cref{fig:hmc} define a set $\clss$ of constraints.  The rule $\ccouple$
generates a coupling constraint for a sampling statement; all other rules
generate standard constrained Horn clauses.

Before walking through the rules, we first set some notation.  We use $\expr_1$
(resp.  $\expr_2$) to denote $\expr$ with all its variables tagged with subscript $1$
(resp.  $2$).  As is standard, we assume there is a one-to-one correspondence
between expressions in our language and our first-order theory, e.g., if
$\bexpr$ is $x > 0$ and $x$ is of type $\mathds{Z}$, then we treat $\bexpr_1$
as the constraint $x_1 >0$ in the theory of linear integer arithmetic. We assume
that the adjacency relation $\adj$ is a formula in our first-order theory.
Finally, technically there are two distinct input variables $\epsilon_1$ and
$\epsilon_2$ representing the target privacy level. Since these variables are
assumed to be equal in any two adjacent states, we will simply use a
single variable $\epsilon$ in the constraints.

Now, we take a closer look at the constraint generation rules. The first two
rules describe the initial and final states: the rule $\cinit$ specifies that
the invariant at $\lentry$ contains all adjacent states and $\costvar$ is $0$.
The rule $\cdp$ states that the invariant at $\lexit$ satisfies the \edp
conditions in \cref{thm:winning} for every return value $\ret$ and every $\eps
> 0$.

The next three rules
describe the coupled postcondition for deterministic statements.
The rule $\cassign$ encodes the effect of executing $v \gets \expr$ in the two
executions of the program. Primed variables, e.g., $v_1'$, to denote the
modified (new) value of $v_1$ after assignment.  The rule $\casmconst$ encodes
effects of assume statements, and $\casmsync$ ensures both processes are
synchronized at assume statements.

The last two rules encode sampling statements.
The rule $\cstrategy$ generates a clause that uses the coupling
encoded by $\strategy_e$ to constrain the values of $v_1'$ and $v_2'$
and increments the privacy cost $\costvar$ by $\distance$.
The rule $\ccouple$ generates a coupling constraint
specifying that $\strategy_e$ encodes a coupling of the
distributions in the two executions.
We interpret the distribution expressions $\dexpr_1,\dexpr_2$ as
distribution families, parameterized by the state.

\begin{example}
For illustration, let us walk through the constraints generated for the simple
program from \cref{ex:cpost}, reproduced below.
Assume the adjacency relation is  $|x_1-x_2| \leq 1 \land \eps > 0$, and let the vector
$\vec{v}$ contain the variables $\{x_1,x_2,\eps,\ret\}$.  We write
$\vec{v}[x']$ for $\vec{v}[x_1\mapsto x_1', x_2\mapsto x_2']$ (as described at
the bottom of \cref{fig:hmc}).  Then, the following constraints
are generated by the indicated rules from \cref{fig:hmc}.
\begin{figure}[h!]
  \vspace{-.2in}
  \begin{minipage}{.3\textwidth}
    \small
    \begin{algorithmic}[1]
      \State $x \gets x + 10$
      \State $x \sim \lap_\eps(x)$
      \State \Return $x$
    \end{algorithmic}
\end{minipage}
\begin{minipage}{.65\textwidth}
  \small
  \begin{align*}
      |x_1 - x_2| \leq 1 \land \eps > 0 \land \costvar=0 \longrightarrow \inv_1(\vec{v},\costvar) && \cinit
     \\
     \inv_1(\vec{v},\costvar) \land x_1' = x_1 + 10 \land x_2' = x_2 + 10 \longrightarrow \inv_2(\vec{v}[x'],\costvar) && \cassign
     \\
     \inv_2(\vec{v}, \costvar) \land \strategy(\vec{v}[x'],x_1',x_2',\distance)
     \longrightarrow
     \inv_3(\vec{v}[x'], \costvar + \distance) && \cstrategy
     \\
     \inv_3(\vec{v},\costvar) \longrightarrow \costvar \leq \eps \land (x_1 = \ret \Rightarrow x_2 = \ret) && \cdp
     \\
     \lap_{\eps}(x_1) \aspace{\strategy(\vec{v},-,-,-)} \lap_{\eps}(x_2)
     && \ccouple
    \end{align*}
\end{minipage}
\end{figure}

\vspace{-.1in}
\noindent
The relation $\strategy$ describes the coupling for the single sampling statement.
\end{example}

\paragraph{Soundness.}
To connect our constraint system to winning coupling strategies, the main
soundness lemma states that a satisfying assignment $\interp$ of $\clss$ encodes
a winning coupling strategy.

\begin{lemma} \label{lem:constr-wcc}
  Let $\interp \models \clss$, where $\clss$ is generated for a program $\prog$
  and adjacency relation $\adj$. Define a family of coupling strategies
  $\{\strp_j\}_{j \in \dom_\retv}$ for $\prog$ by
  \[
    \strp_j(\stmt_e, s_1,s_2)
    = \{(a,a',c) \mid (\vec{q}_j,a,a',c) \in \rho(\strategy_e)\}
  \]
  for every sampling statement $\stmt_e$ in $\prog$ and every pair of states
  $(s_1,s_2)$, where $\vec{q}_j$ replaces each tagged program variable in
  $\vec{v}$ with the value in $s_1$ and $s_2$ respectively, and sets the output
  variable $\ret$ to $j$. Then, $\{\strp_j\}_{j}$ is a winning
  coupling strategy.
\end{lemma}

Now soundness is immediate
by \cref{lem:constr-wcc,thm:winning}.

\begin{theorem}[Soundness of constraints]\label{thm:hmc}
  Let $\clss$ be a system of constraints generated for program $\prog$ and
  adjacency relation $\adj$, as per \cref{fig:hmc}.  If $\clss$ is
  satisfiable, then $\prog$ is $\epsilon$-differentially private.
\end{theorem}
% \begin{proof}
%   By \cref{lem:constr-wcc,thm:winning}.
% \end{proof}

%!TEX root=paper.tex

\section{Parameterized Coupling Families}\label{sec:couplings}
In this section, we
demonstrate how to transform the probabilistic coupling constraints into more
standard logical
constraints. We introduce parameterized families of couplings
(\cref{ssec:couplings}) and then show how to use them to eliminate coupling
constraints (\cref{ssec:transform}).

\subsection{Coupling families} \label{ssec:couplings}

Since coupling strategies select couplings for distribution expressions instead
of purely mathematical distributions, it will be useful to consider couplings
between two families of distributions $M_1$, $M_2$ parametrized by variables,
rather than just between two concrete distributions.
% This will enable us to encode sets of couplings as
% a first-order formula, where each satisfying assignment of the formula denotes
% a coupling.
% The following definition formalizes coupling families.

\begin{definition}[Coupling families]
  Let $M_1, M_2$ be two distribution families with distribution parameters in
  some set $W$. Then, a \emph{coupling family} $\Lambda[w_1, w_2]$ for $M_1$ and
  $M_2$ assigns a coupling
  \[
    M_1[w_1] \aspace{\Lambda[w_1, w_2]} M_2[w_2]
  \]
  for all $(w_1, w_2) \in \pre_\Lambda$,
  where $\pre_\Lambda \subseteq W \times W$ is a \emph{precondition}
  on the parameter combinations.

  We use square brackets to emphasize the distribution parameters.
\end{definition}

% In most cases it is relatively straightforward to generalize a specific coupling
% into a coupling family. However, in some cases we may only be able to find a
% coupling $\Lambda[w_1, w_2]$ when the parameters $w_1, w_2$ are related in some
% way, as a kind of precondition to applying the coupling.

Besides distribution parameters---which model the state of the
distributions---couplings can also depend on logical parameters independent of
the state. We call this second class of parameters \emph{coupling parameters}.
\cref{tbl:couplings} lists the selection of couplings we discuss in this
section along with their preconditions and coupling parameters; the distribution
parameters are indicated by the square brackets.

\paragraph{Null and shift coupling families.}
  In \cref{sec:prelims}, we saw two examples of couplings for the Laplace and
  Exponential distributions: the shift coupling (\cref{def:shift}) and the
  null coupling (\cref{def:null}). These couplings are examples of
  coupling families in that they specify a coupling for a family of
  distributions. For instance, the Laplace distribution has the mean value $z$
  as a distribution parameter; we will write $\lap_\epsilon[z]$ instead of
  $\lap_\epsilon(z)$ when we want to emphasize this dependence.\footnote{%
    Properly speaking $\eps$ should be considered a parameter as well. To
    reduce notation, we will suppress this dependence and treat $\epsilon$ as a
    constant for this section.}
 The null and shift coupling families couple the distribution
 families $\lap_\eps[z_1],
 \lap_\eps[z_2]$ for any two $z_1, z_2 \in \mathds{Z}$.
 For example, for a fixed $z_1, z_2$, we have the shift coupling
 $\lap_\eps(z_1) \aspace{\Lambda_{+k}} \lap_\eps(z_1)$,
 where $\Lambda_{+k} = \{(n_1,n_2, |k+z_1-z_2|\eps) \mid n_1 + k = n_2\}$.
 Notice that $k \in \mathds{Z}$ is a coupling parameter; depending on what we
 set it to, we get different couplings.

 In the case of the exponential distribution, the shift coupling family couples
 the families $\olap_\eps[z_1]$  and $\olap_\eps[z_2]$ under the precondition
 that $z_2 - z_1 \leq k$.

 \begin{table}[t]
   \scriptsize
   \begin{tabular}{cc>{\centering}m{3.4cm}cc}
     \toprule
     Name & Distribution family & Precondition & Coupling family & Coupling params. \\
     \midrule
     Null (Lap) & $\lap_\epsilon[z_1], \lap_\epsilon[z_2]$
     & -- & $\Lambda_\varnothing[z_1, z_2]$ & -- \\
     Shift (Lap) & $\lap_\epsilon[z_1], \lap_\epsilon[z_2]$
     & -- & $\Lambda_{+k}[z_1, z_2]$ & $k \in \mathds{Z}$ \\
     Choice (Lap) & $\lap_\epsilon[z_1], \lap_\epsilon[z_2]$ & $\mathit{NO}(\cpred, \Lambda_{+k}[z_1,z_2], \Lambda_\varnothing[z_1,z_2])$
     & $(\cpred \mathrel{?} \Lambda_{+k}[z_1,z_2]: \Lambda_\varnothing[z_1,z_2])$ & $k \in \mathds{Z}, \cpred \subseteq \mathds{Z}$  \\
     Null (Exp) & $\olap_\epsilon[z_1], \olap_\epsilon[z_2]$
     & -- & $\Lambda_\varnothing[z_1, z_2]$ & -- \\
     Shift (Exp) & $\olap_\epsilon[z_1], \olap_\epsilon[z_2]$
     & $z_2 - z_1 \leq k$ & $\Lambda'_{+k}[z_1, z_2]$ & $k \in \mathds{Z}$ \\
     %Trivial & $\mu_1, \mu_2$ & $\true$ & $\Lambda_\top$ & - \\
     \bottomrule
   \end{tabular}
   \caption{Coupling families (selection)}
   \label{tbl:couplings}
   \vspace{-.15in}
 \end{table}

\paragraph{Choice coupling.}
We will also use a novel coupling construction, the \emph{choice coupling},
which allow us to combine two couplings, using
a predicate on the first sample
space to decide which coupling to apply. If the two couplings
satisfy a \emph{non-overlapping condition}, the result is again a coupling. This
construction is inspired by ideas from \citet{zhang2016autopriv}, who
demonstrate how the pairing between samples (\emph{randomness alignment}) can be
selected depending on the result of the first sample.
These richer couplings can simplify privacy proofs (and hence solutions to
invariants in our constraints), in some cases letting us construct a single
proof that works for all possible output values instead of building a different
proof for each output value.

% We formalize the choice coupling as follows.

\begin{lemma}[Choice coupling] \label{lem:combo-couple}
  Let $\mu_1 \in \sdist(\aset)$, $\mu_2 \in \sdist(\aset)$ be two
  sub-distributions, and let $\cpred$ be a predicate on $\aset$. Suppose that there
  are two couplings
  $
    \mu_1 \aspace{\Lambda} \mu_2
  $
  and
  $
    \mu_1 \aspace{\Lambda'} \mu_2
  $,
  and suppose that $\Lambda, \Lambda'$ satisfy the following
  \emph{non-overlapping} condition: for every pair $a_1 \in \cpred, a_1' \notin
  \cpred$ and $a_2 \in \aset$, $(a_1, a_2, -) \in \Lambda$ and $(a_1', a_2, -)
  \in \Lambda'$ do not both hold.  We abbreviate the non-overlapping condition
  as $\mathit{NO}(\cpred,\Lambda,\Lambda')$.  In other words, there is no
  element $a_2$ that is related under $\Lambda$ to an element in $\cpred$ and
  related under $\Lambda'$ to an element not in $\cpred$. Then the following
  relation
  \[
    (\cpred \mathrel{?} \Lambda : \Lambda') \triangleq
    \{ (a_1, a_2, c) \mid
    (a_1 \in \cpred \implies \Lambda(a_1, a_2, c))
    \land
    (a_1 \notin \cpred \implies \Lambda'(a_1, a_2, c)) \}
  \]
  is a coupling
  $
    \mu_1 \aspace{(\cpred \mathrel{?} \Lambda : \Lambda')} \mu_2 .
  $
\end{lemma}
\begin{proof}
  Let $(\mu_L, \mu_R)$ be witnesses for the coupling $\Lambda$, and let
  $(\mu_L', \mu_R')$ be witnesses for the coupling $\Lambda'$. Then we can
  define witnesses $\nu_L, \nu_R$ for the choice coupling $(\cpred \mathrel{?}
  \Lambda : \Lambda')$:
  \[
    \nu_L(a_1, a_2) \triangleq \begin{cases}
      \mu_L(a_1, a_2) &: a_1 \in \cpred \\
      \mu_L'(a_1, a_2) &: a_1 \notin \cpred
    \end{cases}
    \quad \text{and} \quad
    \nu_R(a_1, a_2) \triangleq \begin{cases}
      \mu_R(a_1, a_2) &: a_1 \in \cpred \\
      \mu_R'(a_1, a_2) &: a_1 \notin \cpred .
    \end{cases}
  \]
  The support and distance conditions follow from the support and distance
  conditions for $(\mu_L, \mu_R)$ and $(\mu_L', \mu_R')$. The first marginal
  condition for $\nu_L$ follows from the first marginal conditions for $\mu_L$
  and $\mu_L'$, while the second marginal condition for $\nu_R$ follows from the
  second marginal conditions for $\mu_R$ and $\mu_R'$ combined with the
  non-overlapping condition on $\cpred$, $\Lambda$, and $\Lambda'$.
\end{proof}

We can generalize a choice coupling to a coupling family.
Suppose we have two coupling families
$\Lambda[.], \Lambda'[.]$ for some distribution families $M_1[.]$, $M_2[.]$,
then we can generate a coupling family $(\cpred, \Lambda[.], \Lambda'[.])$
whose coupling parameters are the predicate $\cpred$ and the parameters
of the two families.
(Note that $\cpred$ could also have parameters.)
In the example below, and in \cref{tbl:couplings}, we give one possible
choice coupling for the Laplace distribution.
An analogous construction
applies for the exponential distribution

% \aws{i think you were mixing up lap and exp; the shift coupling for lap has no
% precondition}
% \jh{Good catch.}
\begin{example}
 Consider the following two
  coupling families for the Laplace distribution family: the shift coupling $\Lambda_{+k}[z_1,z_2]$ and the null coupling
  $\Lambda_{\varnothing}[z_1,z_2]$.
  %If the
  %distribution parameters satisfy the precondition $z_2 - z_1 \geq k$, the
  %
  Then, we have the coupling family
  \[
    \lap_\epsilon[z_1]
    \aspace{(\cpred \mathrel{?} \Lambda_{+k}[z_1,z_2] : \Lambda_\varnothing[z_1,z_2])}
    \lap_\epsilon[z_2] .
  \]
  under the precondition $\mathit{NO}(\cpred, \Lambda_{+k}[z_1,z_2], \Lambda_{\varnothing}[z_1,z_2])$.
  The parameters of this coupling family are the predicate $\cpred \subseteq \mathds{Z}$ and
  the parameter $k \in \mathds{Z}$ of $\Lambda_{+k}[z_1,z_2]$.
  For one possible instantiation, let $\cpred$ be the predicate $ \{ x \in \mathds{Z}
  \mid x \geq 0 \}$ and let $k=1$.  Then if $|z_1 - z_2| \leq 1$,  the
  non-overlapping condition holds and $(\cpred \mathrel{?} \Lambda_{+1}[z_1,z_2] :
  \Lambda_\varnothing[z_1,z_2])$ is a coupling.
\end{example}

% We summarize the coupling family in \cref{fig:couplings}. For each family,
% we give the (possibly parameterized) distributions, the precondition on the
% distribution parameters, and the coupling relation. The final column lists the
% \emph{coupling parameters}---parameters for selecting different couplings. These
% may be simple, like the number $k$ in the shift coupling, or more complicated,
% like the predicate $\cpred$ in the choice coupling.

\subsection{Transforming Horn clauses}\label{ssec:transform}
Now, we can transform \hmc constraints and eliminate coupling constraints by
\emph{restricting} solutions of coupling constraints to use the previous
coupling families. We model coupling parameters and the logic of the coupling
strategy by Horn clauses with \emph{uninterpreted functions}.

\paragraph{Horn clauses with uninterpreted functions.}
Recall our Horn clauses may mention uninterpreted relation symbols $\rels$. We
now assume an additional set $\fs$ of \emph{uninterpreted
  function} symbols, which can appear in interpreted formulas $\varphi$ in
the body/head of a clause $\cls$. A function symbol $\f \in \fs$ of
\emph{arity} $n$ can be applied to a vector of variables of length $n$; for
example, consider the clause
\[
\rel_1(x) \land \f(x,y) = z \longrightarrow \rel_2(z)
\]
In addition to mapping each relation symbol $\rel \in \rels$ to a relation, an
interpretation $\interp$ now also maps each function symbol $\f \in \fs$ to a
function definable in the theory---intuitively, an expression.
We say that $\interp \models \clss$ iff $\interp \cls$ is
valid for all $\cls \in \clss$, where $\interp \cls$ also replaces every
function application $\f(\cvars)$ by $\interp(\f(\cvars))$.

\begin{example}
To give an intuition for the semantics,
consider the two simple clauses:
\begin{align*}
\cls_1 : x = y - 5 \longrightarrow \rel(x,y) &&
\cls_2 : \rel(x,y) \land \f(x) = z \longrightarrow z > y
\end{align*}
A possible satisfying assignment $\interp$
 maps $\rel(x,y)$ to the
formula $x = y - 5$ (ensuring that the first
clause is valid)
and $\f(x)$ to the expression $x + 6$ (ensuring that the
second clause is valid).
\end{example}

Roughly, we can view the problem of solving a set of Horn clauses
$\{C_1,\ldots,C_n\}$ as solving a formula of the form $\exists f_1,\ldots,f_m
\ldotp \exists \rel_1,\ldots,\rel_l \ldotp \bigwedge_i C_i$: Find an
interpretation of $\f_i$ and $\rel_i$ such that $\bigwedge_i C_i$ is valid.  In
our setting, $\f_i$ will correspond to \emph{pieces} of the coupling strategy
and $\rel_i$ will be the invariants showing that the coupling strategy
establishes differential privacy.

\begin{figure}[t!]
  \begin{mdframed}
  \smaller
  \centering
  \begin{align*}
    \trans(\Lambda_\varnothing[z_1,z_2]) &\triangleq d_1 - z_1 = d_2 - z_2 \land \theta = 0\\
    % \trans(\Lambda'_\emptyset[z_1,z_2]) &\triangleq d_1 - z_1 = d_2 - z_2 \land \theta = 0\\
    \trans(\Lambda_{+k}[z_1,z_2]) &\triangleq d_1 + f_k() = d_2 \land \theta = |z_1 - z_2 + f_k()|\cdot \eps\\
    \trans(\cpred \mathrel{?} \Lambda_{+k}[z_1,z_2] : \Lambda_\varnothing[z_1,z_2]) &\triangleq
    (f_b(d_1) \Rightarrow \trans(\Lambda_{+k}[z_1,z_2])) \land
    (\neg f_b(d_1) \Rightarrow \trans(\Lambda_\varnothing[z_1,z_2]))\\
    \trans(\Lambda_{+k}'[z_1,z_2]) &\triangleq d_1 + f_k() = d_2 \land \theta = (z_1 - z_2 + f_k())\cdot \eps\\
    \trans(\pre_{\Lambda_{+k}'[z_1,z_2]}) &\triangleq z_2 - z_1 \leq f_k()\\
    \trans(\pre_{(\cpred \mathrel{?} \Lambda_{+k}[z_1,z_2] : \Lambda_\varnothing[z_1,z_2])})
    &\triangleq
    \left(
     f_b(d_1) \land \neg f_b(d_1') \land \trans(\Lambda_{+k}[z_1,z_2])
    \land \trans(\Lambda_\varnothing[z_1,z_2])'\right) \Rightarrow  d_2 \neq d_2'
  \end{align*}
\end{mdframed}
  \caption{Encodings of coupling families.
  We use $\varphi'$ to denote $\varphi$ with every
  variable $x$ replaced by $x'$.} \label{fig:enc}
\end{figure}

\paragraph{Transforming coupling constraints.}
To transform coupling constraints into Horn clauses with uninterpreted
functions, we encode each coupling family $\Lambda[w_1,w_2]$ as a first-order
formula, $\trans(\Lambda[w_1,w_2])$, along with its precondition.  The encodings
for the families we consider from \cref{tbl:couplings} are shown in
\cref{fig:enc}, where $f_k()$ is a fresh nullary (constant) uninterpreted
function that is shared between the encoding of a family and its precondition.
Similarly, $f_b(\cdot)$ is a fresh uninterpreted function ranging over Booleans,
representing a predicate.

With the encodings of coupling families,
we are ready to define the transformation.
\begin{definition}[Coupling constraint transformer]\label{def:trans}
  Suppose we have a coupling constraint
  \[
  M_1(\cvars_1) \aspace{\rel(\cvars,-,-,-)} M_2[\cvars_2]
  \]
  Suppose we have a set of possible coupling families $\Lambda^1[\cvars_1,\cvars_2],
  \ldots,\Lambda^n[\cvars_1,\cvars_2]$ for $M_1[\cvars_1], M_2[\cvars_2]$.
  We transform the coupling constraint into the following set of Horn clauses,
  where we use $A \leftrightarrow B$ to denote the
  pair of Horn clauses $A \rightarrow B$ and $B \rightarrow A$,
  and we use $f$ to denote a fresh uninterpreted function with range
   $\{1,\ldots,n\}$:
  \begin{align*}
    \left(\bigwedge_{i=1}^n f(\cvars) = i \implies \trans(\Lambda^i[\cvars_1,\cvars_2])\right) \longleftrightarrow \rel(\cvars,d_1,d_2,\theta) \hspace{.6in}\\
    f(\cvars) = 1 \longrightarrow \trans(\pre_{\Lambda^1}[\cvars_1,\cvars_2])
    \qquad\cdots\qquad
    f(\cvars) = n \longrightarrow \trans(\pre_{\Lambda^n}[\cvars_1,\cvars_2])
  \end{align*}
  Without loss of generality, we assume that variables $d_1,d_2,\theta$ do not
  appear in $\cvars, \cvars_1, \cvars_2$.

  Intuitively, the first clause selects which of the $n$ couplings we pick,
  depending on the values $\cvars$.  The remaining clauses ensure that whenever
  we pick a coupling $\Lambda^i$, we satisfy its precondition.
\end{definition}

\begin{example}
Recall our Report Noisy Max example from \cref{sec:example}.
There, we encoded the single sampling statement using the
following coupling constraint:
\[
  \cls_4 :   \lap_{\eps/2}(q_1[i_1]) \aspace{\strategy(\cvars,-,-,-)} \lap_{\eps/2}(q_2[i_2])
\]
Suppose we take the null and shift coupling families for Laplace
as our possible couplings. Our transformation produces:
\[
\bigwedge\begin{array}{l}
f(\cvars) = 1 \implies \left( d_1' - d_2' = q_1[i_1] - q_2[i_2] \land \distance = 0 \right)\\
% \land&&\\
f(\cvars) = 2 \implies \left(d'_1 + f_k() = d'_2 \land \distance = |q_1[i_1]-q_2[i_2] + f_k()|\cdot (\eps/2) \right)
\end{array}
\longleftrightarrow \strategy(\cvars,d_1',d_2',\theta)
\]
We omit clauses enforcing the trivial precondition ($\true$) for these coupling
families.

The first conjunct on the left encodes application of the null
coupling; the second conjunct encodes the shift coupling.
The unknowns are the function $f$ and the constant function $f_k$.
As we sketched in \cref{sec:example}, our tool
discovers the interpretation $\emph{ite}(i_1 \neq \ret, 1, 2)$ for $f(\cvars)$
and the constant $1$ for $f_k()$, where $\emph{ite}(b,a,a')$ denotes
the expression that is $a$ if $b$ is true and $a'$ otherwise.
\end{example}

The following theorem formalizes soundness of our transformation.
\begin{theorem}[Soundness of transformation]\label{thm:trans}
  Let $\clss$ be a set of Horn clauses with coupling constraints.
  Let $\clss'$ be the set $\clss$ with coupling constraints replaced
  by Horn clauses, as in \cref{def:trans}.
  Then, if $\clss'$ is satisfiable, $\clss$ is satisfiable.
\end{theorem}

\section{Implementation and Evaluation}\label{sec:evaluation}

We now describe the algorithms we implement for solving Horn clauses, detail our implementation, and evaluate our technique on well-known
algorithms from the differential privacy literature.

\begin{wrapfigure}{c}{.35\textwidth}
  \smaller
  \centering
  %\vspace{-.3in}
  \begin{algorithmic}[1]
    \Function{solve}{$\clss$}
    \State $n \gets 1$
    \While {$\emph{true}$}
    \State $\interp^{F} \gets \textsc{synth}(\exists F \ldotp \bigwedge \clss^n)$
      \State $\emph{res}  \gets \textsc{verify}(\interp^F \clss)$
      \If {\emph{res} is \abr{SAT}}
        \State \Return \abr{SAT}
      \EndIf
      \State $n \gets n + 1$
    \EndWhile
    \EndFunction
  \end{algorithmic}
  \caption{Algorithm for solving Horn clauses with uninterpreted
  functions}\label{alg:solve}
\end{wrapfigure}

\subsection{Solving Horn clauses with Uninterpreted Functions}\label{ssec:solving}
Solving Horn clauses with uninterpreted functions is
equivalent to a program synthesis/verification problem,
where we have an \emph{incomplete} program with holes
and we want to fill the holes to make
the program satisfy some specification.
% Here, we have Horn clauses with holes---uninterpreted functions and
% relations---and we want to fill the holes such that Horn clauses are
% satisfiable, and therefore imply differential privacy.
In our implementation we employ a
\emph{counterexample-guided inductive synthesis} (\abr{CEGIS})~\citep{DBLP:conf/asplos/Solar-LezamaTBSS06} technique to
propose a solution for the uninterpreted functions $\fs$, and a fixed-point computation with
\emph{predicate abstraction}~\citep{graf1997construction} to verify if the Horn clauses are satisfiable for
the candidate interpretation of $\fs$.  Since these techniques are relatively
standard in the program synthesis literature, we describe our solver at a high
level in this section.
%and defer the details to the
%\iftoggle{long}{Appendix}{supplementary materials}.

\paragraph{Synthesize--Verify Loop.}
Algorithm \textsc{solve} (\cref{alg:solve})
employs a synthesize--verify loop as follows:
\begin{description}
  \item[Coupling Strategy Synthesis.]
  In every iteration of the algorithm,
  the set of clauses $\clss$ are \emph{unrolled}
  into a set of clauses $\clss^n$, a standard technique in Horn-clause solving
  roughly corresponding to unrolling program loops up to a finite bound---see, for example, the algorithm of~\citet{mcmillan2013solving}.
  We assume that there is a single clause $Q \in \clss$
  whose head is not a relation application---we call it
  the \emph{query} clause.
  The unrolling starts with clause $Q$,
  and creates a fresh copy of every clause
  whose head is a relation that appears in the body of $Q$,
  and so on, recursively, up to depth $n$.
  %This is analogous to unrolling program loops $n$ times.
  For an example, consider the following three clauses:
$    r(x) \longrightarrow x > 0$ and
    $r(x) \land f(x) = x' \longrightarrow r(x')$.
  If we unroll these clauses up to $n = 3$, we get the following clauses:
    \rone $r^3(x) \longrightarrow x > 0$,
    \rtwo $r^2(x) \land f(x) = x' \longrightarrow r^3(x')$,
    and \rthree
    $r^1(x) \land f(x) = x' \longrightarrow r^2(x')$,
    where $r^i$ are fresh relation symbols (copies of $r$).

  Once we have non-recursive clauses $\clss^n$,
  we can construct a formula of the form
  $$\exists f_1, \ldots, f_n \ldotp \exists r_1, \ldots, r_m \ldotp \bigwedge_{C \in \clss^n} \cls$$
  where
   $f_i$ and $r_i$ are the function and relation symbols appearing in the unrolled clauses.
  Since $\clss^n$ are non-recursive,
  we can rewrite them to remove the relation
  symbols.
  This is a standard encoding in Horn clause solving
   that is analogous
  to encoding a loop-free program as a formula---%
  as in VC generation and bounded model checking---%
  and we do not detail it here.
  We thus transform $\bigwedge_{C \in \clss^n} \cls$ into
  a constraint of the form
  $\exists f_1, \ldots, f_n \ldotp \forall \vec{x} \ldotp \varphi$
  which can be solved using program synthesis algorithms
  (recall that variables are implicitly universally quantified).
%
  % The set $\clss^n$ is non-recursive,
  % and therefore we can employ
  % loop-free synthesis techniques~\citep{alur2013syntax}
  % to solve the formula $\exists F \ldotp \bigwedge_{\cls \in \clss^n} \cls$.
  % Recall that Horn clauses are universally quantified.
  % To solve such formula,
  Specifically, for \textsc{synth}, we employ
  a version of the symbolic synthesis algorithm due to~\citet{gulwani2011synthesis}.
  The synthesis phase generates an interpretation
  $\interp^F$ of the function symbols,
  a winning coupling strategy that proves \edp for adjacent inputs
  that terminate in $n$ iterations.

  \item[Coupling Strategy Verification.]
  This second phase verifies whether the synthesized
  strategy is indeed a winning coupling strategy for all pairs of adjacent
  inputs by checking if
  $\interp^F \clss$ is satisfiable.
  That
  is, we plug in the synthesized values of $F$ into $\clss$,
  resulting in a set of Horn clauses with no uninterpreted
  functions---only invariant relation symbols $\inv_i$---which can be solved
  with a fixed-point computation.
  The result of \textsc{verify} could be \abr{SAT}, \abr{UNSAT},
  or \abr{UNKNOWN}, due to the undecidability of the problem.
  If the result is \abr{SAT}, then we know that
  $\interp^F$ encodes a winning coupling strategy and we are done.
  Otherwise, we try unrolling further.

  While there are numerous Horn-clause solvers
  available, we have found that they are unable
  to handle the clauses we encountered.
  We thus implemented a custom solver that uses \emph{predicate
  abstraction}.
  We detail our implementation and rationalize this decision below.
\end{description}

\subsection{Implementation Details}

We have implemented our approach
by extending the FairSquare
probabilistic verifier~\citep{albarghouthi17}.
Our implementation takes a program $\prog$ in a simple
probabilistic language with integers and arrays over integers.
An adjacency relation is provided as a first-order
formula over input program variables.
The program $\prog$ is then translated into a
set of Horn clauses over the combined theories
of linear real/integer arithmetic, arrays,
and uninterpreted functions.
% Formulas may be quantified (for instance, to model
% adjacency relations over arrays).

\paragraph{Instantiations of the Synthesizer (\textsc{synth}).}
To find an interpretation of functions $\fs$,
our synthesis algorithm
requires a grammar of expressions to search through.
Constant functions can
take any integer value.
In the case of Boolean and $n$-valued functions,
we instantiate the synthesizer
with a grammar over Boolean operations,
e.g., $\wild_1 > \wild_2$ or $\wild_1 \neq \wild_2$, where $\wild_1$ and
$\wild_2$ can be replaced by variables or numerical
expressions over variables.
In all of our benchmarks, we find that searching
for expressions with \abrs{AST}
of depth 2 is sufficient to finding a winning coupling strategy.

\newcommand{\preds}{\emph{Preds}}
\paragraph{Instantiations of Predicate Abstraction (\textsc{verify}).}
To compute invariants $\inv_i$, we employ
predicate abstraction~\citep{graf1997construction}.
In our initial experiments, we attempted to delegate
this process to existing Horn clause solvers.
Unfortunately, they either diverged
or could not handle the theories we use.\footnote{%
  Duality~\citep{mcmillan2013solving}, which uses \emph{tree interpolants},
  diverged even on the simplest example.  This is due to the interpolation engine
  not being optimized for relational verification; for instance, in all examples
  we require invariants containing multiple predicates of the form $v_1 = v_2$ or
  $|v_1 - v_2| \leq 1$, where $v_1$ and $v_2$ are copies of the same variable.
  SeaHorn~\citep{gurfinkel2015seahorn}, \abr{HSF}~\citep{grebenshchikov2012hsf},
  and Eldarica~\citep{hojjat2012verification} have limited or no support for
arrays, and could not handle our benchmarks.}
We thus built a Horn clause solver
on top of the Z3 \abr{SMT} solver
and instantiated it with a large set of predicates
$\preds$ over a rich set of templates
over pairs of variables.
The templates are of the form:
 $\wild_1 > \wild_2$, $\wild_1 = \wild_2$,
 $\wild_1 + k = \wild_2$,
$|\wild_1 - \wild_2| \leq k$.
By instantiating  $\wild$ with program variables
and $k$ with constants, we generate a large set of predicates.
We also use the predicates appearing in \texttt{assume} statements.
To track the upper bound on $\costvar$,
we use a class of predicates
 of the form $\costvar \leq \sum_{j=1}^n
\wild_j \cdot \iexpr_j$
that model every increment to the privacy cost,
where $\iexpr_j$ is the scale parameter in the $j$th
sampling statement in the program
($\lap_{\iexpr_j}(-)$ or $\olap_{\iexpr_j}(-)$),
and $\wild_j$ is instantiated as an expression
of the cost of the chosen couplings for the $j$th statement,
potentially multiplied by a positive program variable if it
occurs in a loop.
For instance, if we have the shift coupling,
we instantiate $\wild_j$ with $|k + z_1 - z_2|$;
if we use the null coupling, $\wild_j$ is 0.
We use all possible combinations of $\wild_j$
and eliminate any non-linear predicates (see more below).

Since our invariants typically require
disjunctions we employ boolean predicate
abstraction, an expensive procedure requiring exponentially many
\abr{SMT} calls in the size of the predicates in the worst case.
To do so efficiently,
we use the AllSMT algorithm of~\citet{lahiri2006smt}, as implemented
 in the MathSAT5 \abr{SMT} solver~\citep{mathsat}.

\paragraph{Handling Non-linearity.}
Observe that in the definition of shift couplings (\cref{def:shift}),
the parameter $\eps$ is multiplied by an expression,
resulting in a non-linear constraint (a theory
that is not well-supported by \abr{SMT} solvers, due
to its complexity).
To eliminate non-linear constraints, we
make the key observation that it suffices
to set $\eps$ to 1: Since all constraints
involving $\eps$ are summations of coefficients of $\eps$,
we only need to track the coefficients of $\eps$.
%
%\jh{Should make this requirement explicit above.}
%
% We prove soundness of this optimization in the \iftoggle{long}{Appendix}{supplementary materials}.

\paragraph{Handling Conditionals.}
Winning coupling strategies (\cref{thm:winning}) assume synchronization at each conditional statement.
In our implementation, we only impose synchronization at loop
heads by encoding conditional code blocks (with no sampling statements) as monolithic instructions,
as in \emph{large-block encoding}~\citep{beyer2009software}.
This enables handling examples like Report Noisy Max (\cref{sec:example}),
where the two processes may not enter the same branch of the conditional under
the strategy.

\subsection{Differentially Private Algorithms.}

For evaluation, we automatically synthesized privacy proofs for a range of
algorithms from the differential privacy literature. The examples require a
variety of different kinds of privacy proofs. Simpler examples follow by
composition, while more complex examples require more sophisticated arguments.
\cref{fig:algs} shows the code for three of the more interesting
algorithms, which we discuss below; \cref{tbl:bench} presents the full
collection of examples.

\begin{figure}[t!]
  \smaller
  \centering
  \begin{minipage}[t]{.35\textwidth}
    \begin{algorithmic}
      % \State  $\adj: \forall j \in [0,|q_1|) \ldotp |q_1[j] - q_2[j]| \leq 1$
      %   $\land |q_1| = |q_2|$
      % \State
      \Function{\sc SmartSum}{$q$}
      \State $\mathit{next}, n, i, sum \gets 0$
      \State $r \gets []$
      \While{$i < |q|$}
      \State $\mathit{sum} \gets \mathit{sum} + q[i]$
      \If{$(i + 1) \mathbin{\%} M = 0$}
      \State $n \sim \lap_\eps(n + \mathit{sum})$
      \State $\mathit{sum}, \mathit{next} \gets 0, n$
      \Else
      \State $\mathit{next} \sim \lap_\eps(\mathit{next} + q[i])$
      \EndIf
      \State $r \gets \mathit{next} :: r$
      \State $i \gets i + 1$
      \EndWhile
      \State \Return $r$
      \EndFunction
    \end{algorithmic}
  \end{minipage}%
  \vrule
  \begin{minipage}[t]{.3\textwidth}
    \begin{algorithmic}
      % \State
      % $\adj: \forall j \in [1,R] \ldotp |\textsf{qu}(d_1,j) - \textsf{qu}(d_2,j)| \leq 1$
      % \State
      \Function{\sc ExpMech}{$d,\mathit{qscore}$}
      \State $i \gets 1$
      \State $\mathit{bq} \gets 0$
      \While {$i \leq R$}
      \State $s \gets \mathit{qscore}(i, d)$
      \State $\mathit{cq} \gets \olap_{\eps/2}(s)$
      \If {$\mathit{cq} > \mathit{bq} \lor i = 1$}
      \State $\mathit{max} \gets i$
      \State $\mathit{bq} \gets \mathit{cq}$
      \EndIf
      \State $i \gets i + 1$
      \EndWhile
      \State \Return $\mathit{max}$
      \EndFunction
    \end{algorithmic}
  \end{minipage}
  \vrule
  \begin{minipage}[t]{.3\textwidth}
    \begin{algorithmic}
      % \State $\adj: \forall j \in [0,|q_1|) \ldotp |q_1[j] - q_2[j]| \leq 1$
      %   \State    $\land |q_1| = |q_2| \land N_1 = N_2$
      % \State
      \Function{\sc NumericSparseN}{$\mathit{qs},T,N$}
      \State $r \gets []$
      \State $i, \mathit{ct} \gets 0$
      \State $t \sim \lap_{\eps/3}(T)$
      \While {$i < |\mathit{qs}| \land \mathit{ct} < N$}
      \State $n \sim \lap_{\eps/6N}(\mathit{qs}[i])$
      \If {$n > t$}
      \State $\mathit{ans} \sim \lap_{\eps/3N}(\mathit{qs}[i])$
      \State $r \gets (i, \mathit{ans}) :: r$
      \State $\mathit{ct} \gets \mathit{ct} + 1$
      \EndIf
      \State $i \gets i + 1$
      \EndWhile
      \State \Return $r$
      \EndFunction
    \end{algorithmic}
  \end{minipage}%
  \caption{Three representative benchmarks (in practice,
  lists are encoded as arrays)}
  \label{fig:algs}
  \vspace{-.2in}
\end{figure}

% \aws{i changed the descriptions to talk more
% about how these are used in practice}

\paragraph{Two-Level Counter (\textsc{SmartSum}).}
The \textsc{SmartSum} algorithm~\citep{CSS10,DNPR10} was designed to continually
publish aggregate statistics while maintaining privacy---for example,
continually releasing the total number of visitors to a website.
Suppose that we have a list of inputs $q$,
where $q[i]$ denotes the total number of visitors in hour $i$.
\textsc{SmartSum} releases a list of all running sums:
$q[0]$, $q[0] + q[1], \ldots, \sum_{j=0}^i q[j], \ldots$.
Rather than adding separate noise to each $q[i]$ or adding separate noise to
each running sum---which would provide weak privacy guarantees or inaccurate
results---\textsc{SmartSum} chunks the inputs into blocks of size $M$ and adds
noise to the sum of each block. Then, each noisy running sum can be computed by
summing several noised blocks and additional noise for the remaining inputs. By
choosing the block size carefully, this approach releases all running sums with
less noise than more na\"ive approaches, while ensuring \edp.
%Intuitively, privacy follows since the differing
%input $q[i]$ by 1 can affect both noising steps: when the algorithm first
%encounters $i$, and when the algorithm noises the block containing $i$.

\paragraph{Discrete Exponential Mechanism (\textsc{ExpMech}).}
The exponential mechanism~\citep{MT07}
is used when \rone the output of an algorithm
is non-numeric or \rtwo different outputs
have different \emph{utility}, and
adding numeric noise to the output would produce an unusable answer,
e.g., in the case of an auction where
we want to protect privacy of bidders~\citep{dwork2014algorithmic}.
\textsc{ExpMech} takes a database and a \emph{quality
  score} (utility function) as input, where the quality score maps the database and each element of
the range to a numeric score. Then, the algorithm releases the element
with approximately the highest score. We verified a version of this algorithm
that adds noise drawn from the exponential mechanism to each quality score, and then
releases the element with the highest noisy score. This implementation is also
called the \emph{one-sided Noisy Arg-Max}
\citep{dwork2014algorithmic}.

% The privacy proof for this algorithm is not straightforward---sampling from the
% exponential mechanism does not satisfy differential privacy, so we cannot prove
% privacy from composition. Instead, we can give a coupling proof of privacy
% similar to what we saw for the algorithm Report Noisy Max (\cref{sec:example}).

\paragraph{Sparse Vector Mechanism (\textsc{NumericSparseN}).}
The Sparse Vector mechanism (also called \textsc{NumericSparse} in the textbook by \citet{dwork2014algorithmic})
is used in scenarios where we would like to answer
a large number of numeric queries while only paying for
queries with large answers.
To achieve \edp, \textsc{NumericSparseN} releases
the noised values of the first $N$ queries that
are above some known threshold $T$,
while only reporting that other queries
are below the threshold, without disclosing their values.

Specifically, \textsc{NumericSparseN}
takes a list of numeric queries
$\mathit{qs}$ and numeric threshold $T$. Each query is assumed to be
$1$-sensitive---its
answer may differ by at most $1$ on adjacent databases. The threshold
$T$ is assumed to be public knowledge, so we can model it as being equal in
adjacent inputs. The Sparse Vector mechanism adds noise to the threshold and then computes
a noisy answer for each query. When the algorithm finds a query where the noisy
answer is larger than the noisy threshold, it records the index of the query. It
also adds fresh noise to the query answer, and records the (freshly) noised
answer as well. When the algorithm records $N$ queries or runs out of queries,
it returns the final list of indices and answers.

The privacy proof is interesting for two reasons. First, applying the
composition theorem of privacy does prove privacy, but with an overly
conservative level of $\epsilon$ that depends on the number of queries. A more
careful proof by coupling shows that the privacy level depends only on the
number of above threshold queries, potentially a large saving if most of the
queries are below threshold. Second, adding \emph{fresh} noise when estimating
the answer of above threshold queries is critical for privacy---it is not
private to reuse the noisy answer of the query from checking against the noisy
threshold. Previous versions of this algorithm suffered from this
flaw~\citep{lyu2016understanding}.

\begin{table}[t!]
  \caption{Differentially private algorithms used in our evaluation}
  \label{tbl:bench}
  \smaller
  \begin{tabular}{lp{0.75\textwidth-2\tabcolsep}}
    \toprule
    Algorithm & Description\\
    \midrule
    \textsc{PartialSum}
    & Compute the noisy sum of a list of queries. \\
    \textsc{PrefixSum}
    & Compute the noisy sum for every prefix of a list of queries. \\
    \textsc{SmartSum}
    & Advanced version of \textsc{PrefixSum} that chunks the list
    \citep{DNPR10,CSS10}. \\
    \textsc{ReportNoisyMax}
    & Find the element with the highest quality score
    \citep{dwork2014algorithmic}. \\
    \textsc{ExpMech}
    & Variant of \textsc{ReportNoisyMax} using the exponential distribution
    \citep{dwork2014algorithmic,MT07}. \\
    \textsc{AboveThreshold}
    & Find the index of the first query above threshold
    \citep{dwork2014algorithmic}. \\
    \textsc{AboveThresholdN}
    & Find the indices of the first $N$ queries with answer above threshold
    \citep{dwork2014algorithmic,lyu2016understanding}. \\
    \textsc{NumericSparse}
    & Return the index and answer of the first query above threshold
    \citep{dwork2014algorithmic}. \\
    \textsc{NumericSparseN}
    & Return the indices and answers of the first $N$ queries above threshold
    \citep{dwork2014algorithmic,lyu2016understanding}. \\
    \bottomrule
  \end{tabular}
\end{table}

\subsection{Experimental Results}

\cref{tbl:results} summarizes the results of  applying our
implementation to the differentially private algorithms described in
\cref{tbl:bench}.  For each algorithm, we established \edp; in case of
\textsc{SmartSum}, we show that it is $2\eps$-\dpriv, as established
by~\citet{CSS10}.\footnote{%
  Technically, this requires changing the upper bound on $\costvar$ from $\eps$
  to $2\eps$ in the clause generated by rule $\cdp$ in \cref{fig:hmc}.}
In all our examples, we transform
coupling constraints (using \cref{def:trans})
by instantiating $\strategy_e$ (for every sampling sampling
statement $\stmt_e$) with two choices, the null
and shift coupling families of $\lap$ or $\olap$---i.e.,
the function $f(\vec{v})$ in \cref{def:trans}
has range $\{1,2\}$.
While this is sufficient to find a proof
for all examples, in the case of the
Sparse Vector mechanism (\textsc{NumericSparse*}),
we found that the proof requires an invariant
that lies outside our decidable theory.
Using the choice coupling family for those algorithms, our tool can automatically
discover a much simpler invariant.

\paragraph{Results Overview.}
For each algorithm, \cref{tbl:results} shows
the privacy bound and the
number of sampling statements (\#Samples) appearing
in the algorithm to give a rough idea of how many couplings must be found.
The verification statistics show
the number of predicates (\#Preds)
generated from our predicate templates, and
the time spent in \textsc{verify}.
The synthesis statistics show the size (in terms of \abr{ast} nodes)
of the largest formula passed to \textsc{synth}
and number of variables
appearing in it;
the number of \abr{CEGIS} iterations of the synthesis
algorithm~\citep{gulwani2011synthesis};
and the time spent in \textsc{synth}.

In all cases, our implementation was able to automatically
establish differential privacy in a matter of minutes.
However, there is significant variation in the time needed
for different algorithms.
Consider, for instance, \textsc{NumericSparseN}.
While it is a more general version of \textsc{AboveThresholdN},
it requires less verification time.
This is because the invariant required by \textsc{NumericSparseN}
is conjunctive, allowing \textsc{verify} to quickly
discover an invariant.
The invariant required by \textsc{AboveThresholdN}
contains multiple disjuncts, making it spend more time in boolean abstraction.
(Note that an upper bound on the number of possible disjunctive invariants
using $n$ predicates is $2^{2^n}$.)
Synthesis time, however, is higher for \textsc{NumericSparseN}.
This is due to the use of the choice coupling family, which
has two coupling parameters and requires discovering a
non-overlapping predicate $P$.

\paragraph{Detailed Discussion.}
For a more detailed view, let us consider
the \textsc{ReportNoisyMax} algorithm. In \cref{sec:example},
we described in detail the winning coupling strategy our tool discovers.
To verify that the coupling strategy is indeed winning,
the example requires the following coupled invariant:
\begin{align*}
  \mathcal{I} &\triangleq (i_1 = i_2)
            \land (i_1 \leq \ret \implies \mathcal{I}_{\leq} \land \costvar \leq 0)
            \land (i_1 > \ret \implies \mathcal{I}_> \land \costvar \leq \eps),
\end{align*}
where
\begin{align*}
  \mathcal{I}_\leq &\triangleq ((r_1 = r_2 = \ret = 0) \lor (r_1 < \ret \land r_2 < \ret))
            \land |\emph{best}_1 - \emph{best}_2| \leq 1 \\
  \mathcal{I}_> &\triangleq r_1 = \ret \Longrightarrow (r_2 = \ret \land \emph{best}_1 + 1 = \emph{best}_2)
  .
\end{align*}
At a high-level, the invariant considers two cases: \rone when the loop
counter $i_1$ is less than or equal to the output $\ret$, and \rtwo when the
loop counter $i_1$ has moved past the output $\ret$. In the latter
case, the invariant ensures that if $r_1=\ret$, then $r_2=\ret$, and, crucially,
$\emph{best}_1 + 1 = \emph{best}_2$.  Notice also how the invariant establishes
that the privacy cost is upper-bounded by $\eps$.  Our tool discovers a similar
invariant for \textsc{ExpMech}, with the strategy also establishing the
precondition of the shift coupling.

Let us now consider a more complex example,
the \textsc{NumericSparseVectorN}, which utilizes
the choice coupling.
Notice that there are 3 sampling statements in this example.
\begin{itemize}
\item For the first sampling statement, we use the shift
coupling to ensure that $t_1 + 1 = t_2$; since
$T_1 = T_2$, we incur a privacy cost of $\eps/3$.
\item For the second sampling statement, we use the choice coupling:
When $n_1 > t_1$, we use the shift coupling
to ensure that $n_1 + 1  = n_2$;
otherwise, we take the null coupling.
Since $\emph{qs}_1[i_1]$ and $\emph{qs}_2[i_2]$
may differ by one, as per the adjacency relation,
we will incur a worst-case price of $\eps/3$ for this statement,
since we will incur $2*\eps/6N$ cost for each of the $N$ iterations
where $n_1 > t_1$.
\item Finally, for the third sampling statement,
use the shift coupling to ensure that $\emph{ans}_1 = \emph{ans}_2$,
incurring a cost of $\eps/3$ across the $N$ iterations where the
conditional is entered.
\end{itemize}
As such, this strategy establishes
that the \textsc{NumericSparseVecN} is $\eps = \eps/3 + \eps/3 + \eps/3$
differentially private.
It is easy to see that this is a winning coupling strategy,
since whenever the first process enters the conditional
and updates $r$, the second process also enters the conditional
and updates $r$ with the same value---as the strategy
ensures that $\emph{ans}_1 = \emph{ans}_2$ and $i_1 = i_2$.

\paragraph{Summary.}
Our results demonstrate that our proof technique based on winning
coupling strategies is
\rone amenable to automation through Horn-clause solving, and
\rtwo is applicable to non-trivial algorithms proposed in the differential
privacy literature.
To the best of our knowledge, our technique is the first to automatically
establish privacy for all of the algorithms in \cref{tbl:results}.

\begin{table}[t!]
  \caption{Experimental results (OS X 10.11; 4GHz Intel Core i7; 16GB RAM)}
  \label{tbl:results}
  \smaller
  \begin{tabular}{lllllllll}
    \toprule
    \multicolumn{3}{c}{} & \multicolumn{2}{c}{Verification stats.} & \multicolumn{4}{c}{Synthesis stats.} \\
    \cmidrule(lr){4-5}
    \cmidrule(lr){6-9}
    Algorithm &  Bound & \#Samples & \#\preds & Time (s)  & \#vars & Formula size & \#\abr{CEGIS} iters. & Time (s) \\
    \midrule
    \textsc{PartialSum}    &  $\eps$   & 1 & 30 & 12 & 183 & 566 & 9 & 0.4 \\
    \textsc{PrefixSum}    &  $\eps$   & 1 & 40 & 14 & 452 & 1246 & 8 & 1.1 \\
    \textsc{SmartSum}    &  $2\eps$   & 2 & 44 & 255 & 764 & 2230 & 1766 & 579.2 \\
    \textsc{ReportNoisyMax}    &  $\eps$   & 1 & 36 & 22 & 327 & 1058 & 35 & 1.5 \\
    \textsc{ExpMech}    &  $\eps$   & 1 & 36 & 22 & 392 & 1200 & 152 & 5.0 \\
    \textsc{AboveThreshold}    &  $\eps$   & 2 & 37 & 27 & 437 & 1245 & 230 & 7.3 \\
    \textsc{AboveThresholdN}    &  $\eps$   & 3 & 59 & 580 & 692 & 1914 & 628 & 31.3 \\
    \textsc{NumericSparse}    &  $\eps$   & 2 & 62 & 4 & 480 & 1446 & 65 & 3.2 \\
    \textsc{NumericSparseN}    &  $\eps$   & 3 & 68 & 5 & 663 & 1958 & 6353 & 1378.9 \\
    \bottomrule
  \end{tabular}
\end{table}

%!TEX root=paper.tex

\section{Related Work}\label{sec:related}

\paragraph{Formal Verification of Differential Privacy.}
%
% Soon after \citet{DMNS06} defined differential privacy,
% \citet{pinq} leveraged
% the sequential composition principle to design \abr{PINQ}, a library for constructing
% differentially private queries with a dynamic analysis for tracking the privacy
% level. Subsequently,
Researchers have explored a broad array of techniques for static verification of
differential privacy, including linear and dependent type
systems~\citep{ReedPierce10,GHHNP13,BGGHRS15,AGGH14,ACGHK16,zhang2016autopriv},
relational program
logics~\citep{BKOZ13-toplas,BartheO13,OlmedoThesis,BGGHS16,BGGHS16c,Sato16,BEHSS17,JHThesis},
partial evaluation~\citep{adafuzz}, and more. \citet{Murawski:2016:2893582}
provide a recent survey.

Our work is inspired by the LightDP
system~\citep{zhang2016autopriv}, which combines a relational type system
for an imperative language, a novel type-inference algorithm, and a product
program construction. Types in the relational type system encode randomness
alignments. Numeric types have the form $\mathtt{num}_{\text{d}}$, where
$\text{d}$ is a \emph{distance expression} that can mention program variables.
Roughly speaking, the distance $\text{d}$ describes how to map samples $x$ from
the first execution to samples $x + |\text{d}|$ in the second execution.
\citet{zhang2016autopriv} focus on randomness alignments that are injective maps
from samples to samples. Such maps also give rise to approximate couplings; for
instance, if $g : \mathds{Z} \to \mathds{Z}$ is injective, then there exists a
variable approximate coupling of two Laplace distributions $\Lambda = \{ (z,
g(z), c(z)) \}$ for some costs $c(z)$. However, approximate couplings are more
general than randomness alignments, as couplings do not require an injective map
on samples.

\citet{zhang2016autopriv} also show how distance expressions (and hence
randomness alignments) in LightDP can be combined: they can be added and
subtracted ($\text{d}_1 \oplus \text{d}_2$), multiplied and divided ($\text{d}_1
\otimes \text{d}_2$), and formed into conditional expressions ($\text{d}_1 \odot
\text{d}_2 \mathrel{?} \text{d}_3 : \text{d}_4$, where $\odot$ is a binary
comparison). This last conditional construction is highly useful for proving
privacy of certain examples, and is modeled by the choice coupling in our system
(\cref{sec:couplings}).  Conceptually, our work gives a better understanding of
randomness alignment and approximate couplings---in some sense, randomness
alignments can be seen as a particular case of (variable) approximate couplings.
LightDP also uses MaxSMT to find an optimal privacy guarantee.
Our approach is property-directed: we try to prove a set upper-bound.
We could enrich the generated constraints with an objective
function minimizing the privacy cost $\costvar$, and solve them
using an optimal synthesis technique~\citep{bornholt2016optimizing}.

We also draw on ideas from \citet{BGGHKS14}, who,
like~\citet{zhang2016autopriv}, use a product construction for proving
differential privacy.
Each call to a sampling instruction is replaced by a
call to a non-deterministic function, incrementing the cost in a ghost variable
$v_\epsilon$. However, their system cannot prove privacy beyond composition, and
the cost is required to be deterministic. In particular, many of the advanced
examples that we consider, like Report Noisy Max and the Sparse Vector
mechanisms, are not verifiable. Our work can be seen as extending their system
along three axes:
\rone supporting richer couplings for the sampling instructions to handle more
complex examples,
\rtwo reasoning about randomized privacy costs, and
\rthree automatically constructing the proof.

\paragraph{Automated Verification and Synthesis.}
We reduced the problem of verifying \edp to solving Horn clauses with
uninterpreted functions. \citet{beyene2013solving} introduced existentially
quantified Horn clauses, which were used for infinite-state \abr{CTL} verification and
solving $\omega$-regular games~\citep{Beyene14}. There, variables in the head of
a clause can be existentially quantified. To solve the clauses, they
\emph{Skolemize} the existential variables and discover a solution to a Skolem
relation. Using this technique, we could have adopted an alternative encoding by
using existentially quantified variables to stand for unknown couplings in
strategies. However, for correctness, we would need to restrict solutions for Skolem relations to
be couplings---e.g., by enforcing a template for the Skolem relations.

Solving Horn clauses with uninterpreted functions is closely connected to
\emph{sketching-based} program synthesis
problems~\citep{DBLP:conf/asplos/Solar-LezamaTBSS06}, where we want to find substitutions for
holes in a program to ensure that the whole program satisfies a target property.
One could reformulate our problem of solving Horn clauses with uninterpreted functions
as solving a program synthesis problem, and, e.g., use the template-based
technique of~\citet{Srivastava10} to find an inductive invariant and a coupling
strategy. In our implementation, we used a synthesize-and-verify loop, a
methodology that has appeared in different guises, for example, in the Sketch
synthesizer~\citep{DBLP:conf/asplos/Solar-LezamaTBSS06} and the \abr{E-HSF} Horn-clause
solver~\citep{beyene2013solving}.

\paragraph{Hyperproperties.}
There is a rich body of work on verifying relational/hyper- properties of
programs, e.g., in
translation-validation~\citep{pnueli1998translation,necula2000translation},
security
problems~\citep{clarkson2010hyperproperties,terauchi2005secure}, and
Java code~\citep{Sousa16}. There has been far less work on automated reasoning
for probabilistic relational properties. Techniques for reasoning about
quantitative information flow~\citep{kopf2010approximation} typically reason
about probabilities, since the property depends on a distribution on inputs.
However, since there are no universally quantified variables in the property, the
verification problem can be handled by \emph{model-counting}: counting the
number of executions that satisfy certain properties. The \edp property
quantifies over all sets of adjacent inputs and all outputs, making automated
verification significantly more challenging.

%!TEX root=paper.tex

\section{Conclusions and Future Work}\label{sec:conclusion}
We have presented a novel proof technique for $\eps$-differential privacy,
by finding a strategy that uses variable
approximate couplings to pair two adjacent executions of a program.
We formulated the set of \emph{winning} strategies in a novel
constraint system using Horn clauses and coupling constraints.
By carefully restricting solutions of coupling constraints to
encodings of coupling families, we can automatically
prove correctness of complex algorithms from the differential privacy literature.

The next natural step would be to automate proofs of $(\epsilon,
\delta)$-differential privacy. Approximate couplings have been used in the past
to verify this more subtle version of privacy, but there are both conceptual and
practical challenges in automatically finding these proofs. On the
theoretical end, while we proved $\epsilon$-DP by finding a coupling strategy
with cost $\epsilon$ for each output, if we find a coupling strategy with cost
$(\epsilon, \delta)$ for each output, then $\delta$ parameters will sum up over
all outputs to give the privacy level for all outputs---this can lead to a very
weak guarantee, or a useless guarantee if there are infinitely many outputs. On
the more practical side, in many cases the $\delta$ parameter in an $(\epsilon,
\delta)$-private algorithm arises from an \emph{accuracy bound} stating that a
certain key sample has a small, $\delta$ probability of being too large. These
bounds often have a complex form, involving division and logarithm operations.
Due to these and other obstacles, to date there is no automatic system for
proving $(\epsilon, \delta)$-privacy. Nevertheless, extending our techniques to
find such proofs would be an intriguing direction for future work.

There are other possible targets of our approach beyond differential privacy.
For example, we plan to adapt our technique to synthesize independence and
uniformity for properties of probabilistic programs~\citep{barthe2017proving}.
Additionally, we plan to combine our techniques with probabilistic resource
bounds analyses to synthesize couplings that provide upper bounds on distances
between probabilistic processes, e.g., to show rapid mixing of Markov
chains~\citep{barthe2017coupling}.

% Acknowledgments
\begin{acks}                            %% acks environment is optional
                                        %% contents suppressed with 'anonymous'
  %% Commands \grantsponsor{<sponsorID>}{<name>}{<url>} and
  %% \grantnum[<url>]{<sponsorID>}{<number>} should be used to
  %% acknowledge financial support and will be used by metadata
  %% extraction tools.
  % This material is based upon work supported by the
  % \grantsponsor{GS100000001}{National Science
  %   Foundation}{http://dx.doi.org/10.13039/100000001} under Grant
  % No.~\grantnum{GS100000001}{nnnnnnn} and Grant
  % No.~\grantnum{GS100000001}{mmmmmmm}.  Any opinions, findings, and
  % conclusions or recommendations expressed in this material are those
  % of the author and do not necessarily reflect the views of the
  % National Science Foundation.
  We thank Gilles Barthe, Marco Gaboardi, Zachary Kinkaid, Danfeng Zhang, and
  the anonymous reviewers for stimulating discussions and useful comments on
  earlier drafts of this work.
  This work is partially supported by the \grantsponsor{GS100000144}{National Science Foundation CNS}{http://dx.doi.org/10.13039/100000144}
  under Grant No.~\grantnum{GS100000144}{1513694},
  the \grantsponsor{GS100000143}{National Science Foundation CCF}{http://dx.doi.org/10.13039/100000143}
  under Grant Nos.~\grantnum{GS100000143}{1566015}, \grantnum{GS100000143}{1704117}, \grantnum{GS100000143}{1652140},
  the \grantsponsor{GS100010663}{European Research Council}{http://dx.doi.org/10.13039/100010663}
  under Grant No.~\grantnum{GS100010663}{679127},
  and the \grantsponsor{GS100000893}{Simons Foundation}{http://dx.doi.org/10.13039/100000893}
  under Grant No.~\grantnum{GS100000893}{360368} to Justin Hsu.
\end{acks}

%% Bibliography
\bibliography{header,refs}

\iftoggle{long}{\appendix\section{Omitted proofs} \label{app:proofs}

\note{some theorem statements have moved -- perhaps we should just
refer to theorem numbers instead of restating}

\subsection{Proofs for variable approximate couplings}
The following proves \cref{lemma:dpcoupling},
which connects variable approximate couplings with \edp.

\begin{proof}
  This is a standard lemma about approximate couplings; see, e.g.
  \citet[Proposition 6]{BGGHS16}. We provide a self-contained proof here. Let
  $j$ be any element of $\dom_\retv$ and let $(s_1, s_2) \in \adj$ be two
  adjacent inputs. We have two witness distributions $\mu_L, \mu_R$ to the
  approximate coupling $\sem{P}(s_1) \aspace{\Lambda_j} \sem{P}(s_2)$. Therefore, we can
  bound
  \begin{align}
    \sem{P}(s_1)( \{ s_1' \mid s_1'(\retv) = j \} )
    &= \mu_L( \{ (s_1', s_2') \mid s_1'(\retv) = j \} )
    \tag{First marginal} \\
    &\leq \exp(\epsilon) \cdot \mu_R( \{ (s_1', s_2') \mid s_1'(\retv) = j \} )
    \tag{Distance} \\
    &\leq \exp(\epsilon) \cdot \mu_R( \{ (s_1', s_2') \mid s_2'(\retv) = j \} )
    \tag{Support} \\
    &\leq \exp(\epsilon) \cdot \sem{P}(s_2)( \{ s_2' \mid s_2'(\retv) = j \} ) .
    \tag{Second marginal}
  \end{align}
  We use $\eps$ for $s_1(\eps)$.

  % \aws{I don't get the "support" step. How do you flip $s_1'$ to $s_2'$.}

  % \jh{In everything with non-zero probability, $s_1'(\retv) = j$ implies that
  %   $s_2'(\retv) = j$. So every pair $(s_1', s_2')$ satisfying $s_1'(\retv) = j$
  %   must also satisfy $s_2'(\retv) = j$, so the set with $s_2'(\retv) = j$ is a
  %   superset/has larger probability. Makes sense?}
\end{proof}

\subsection{Proofs for coupling strategies}
\cref{lemma:strat-couple} is the key technical result connecting coupling
strategies to couplings. To build the witness, we will need a few standard
distribution operations.

\begin{definition}
  Let $\aset, \aset'$ be discrete sets.
  \begin{itemize}
    \item The \emph{zero sub-distribution} $\dzero : \sdist(\aset)$ assigns
      weight $0$ to all elements.
    \item The \emph{unit} map $\dunit : \aset \to \sdist(\aset)$ is defined by
      $\dunit(a)(a') \triangleq \mathds{1}[a = a']$.
    \item The \emph{bind} map $\dbind : \sdist(\aset) \times (\aset \to
      \sdist(\aset')) \to \sdist(\aset')$ is defined by:
      \[
        \dbind(\mu, f)(a') \triangleq \sum_{a \in \aset} \mu(a) \cdot f(a)(a') .
      \]
  \end{itemize}
\end{definition}

We will use the following lemma about bind and projections.

\begin{lemma} \label{l:comp-marg-supp}
  Suppose we have functions $f : \aset \to \sdist(\aset')$ and $g : \aset' \to
  \sdist(\aset'')$.  Let $\Phi \subseteq \aset \times \aset, \Phi' \subseteq
  \aset' \times \aset', \Phi'' \subseteq \aset'' \times \aset''$ be binary
  relations. Let $i = 1$ or $2$. Suppose that we have maps $F : \aset \times
  \aset \to \sdist(\aset' \times \aset')$ and $G : \aset' \times \aset' \to
  \sdist(\aset'' \times \aset'')$ such that
  \begin{itemize}
    \item $\pi_i(F(a_1, a_2)) = f(\pi_i(a_1, a_2))$ and $\supp(F(a_1, a_2)) \subseteq \Phi'$
      for every $(a_1, a_2) \in \Phi$; and
  \end{itemize}
  \begin{itemize}
    \item $\pi_i(G(a_1', a_2')) = g(\pi_i(a_1', a_2'))$ and $\supp(G(a_1', a_2')) \subseteq \Phi''$
    for every $(a_1', a_2') \in \Phi'$.
  \end{itemize}
  Then for every $(a_1, a_2) \in \Phi$, we have
  \begin{itemize}
    \item $\pi_i(\dbind(F(a_1, a_2), G)) = \dbind(f(\pi_i(a_1, a_2)), g)$
      and $\supp(\dbind(F(a_1, a_2), G)) \subseteq \Phi''$ .
  \end{itemize}
\end{lemma}
\begin{proof}
  By unfolding definitions.
\end{proof}

We can now prove \cref{lemma:strat-couple}.

\begin{lemma}[From strategies to couplings (\cref{lemma:strat-couple})]
Suppose $\strp$ is synchronizing for $\trace$ and $(s_0, s_0') \in \adj$.  Let
$\Psi = \cpost_\strp( \{ (s_0, s_0', 0) \}, \sigma )$, and let $f : \states
\times \states \to \mathds{R}$ be such that $c \leq f(s, s')$ for all $(s, s',
c) \in \Psi$.  Then we have a coupling
\[
  \sigma(s_0) \aspace{\Psi_f} \sigma(s_0')
\]
with $\Psi_f \triangleq \{ (s, s', f(s, s')) \mid (s, s', -) \in \Psi \}$.
\end{lemma}
\begin{proof}
  We will construct witnesses $\mu_L, \mu_R$ for the desired coupling by
  induction on the length of $\sigma$. In the base case $\sigma$ is the empty
  command and the claim is straightforward: $\mu_L = \mu_R \triangleq
  \dunit(s_0, s_0')$ witness the approximate coupling with cost $f = 0$.

  In the inductive case, $\sigma = \sigma' ; \stmt$ for a command $\stmt$. Let
  $\Psi \triangleq \cpost_\strp( \{ (s_0, s_0', 0) \}, \sigma)$ be the coupled
  post-condition for the whole trace, and let $\Theta \triangleq \cpost_\strp(
  \{ (s_0, s_0', 0) \}, \sigma)$ be the coupled postcondition for the prefix. We
  will write $\domain{\Psi}$ and $\domain{\Theta}$ for the coupled
  postconditions projected to the first two components.

  If the final statement $\stmt$ is an assignment $x \gets \expr$, define a cost
  function $\overline{f} : \states \times \states \to \mathds{R}$ by
  \[
    \overline{f}(s_i, s_i')
    \triangleq f(s_i[x \mapsto s_i(\expr)], s_i'[x \mapsto s_i'(\expr)]) .
  \]
  By definition of $\Psi$, we must have $c \leq \overline{f}(s_i, s_i')$ for
  every $(s_i, s_i', c) \in \Theta$. Hence by induction, we have an approximate
  coupling
  \[
    \sem{\sigma'} s_0 \aspace{\Psi_{\overline{f}}} \sem{\sigma'} s_0'
  \]
  where $\Psi_{\overline{f}} \triangleq \{ (s, s', \overline{f}(s, s')) \mid
  (s, s') \in \domain{\Theta} \}$. So every pair of inputs $(s_0, s_0') \in
  \Delta$ gives a pair of witness distributions $\mu_L', \mu_R'$, in particular
  they satisfy the marginal and support conditions.

  Now take any two states $(s_i, s_i') \in \domain{\Theta}$ and define
  \[
    w_L(s_i, s_i') = w_R(s_i, s_i')
    \triangleq \dunit(s_i[x \mapsto s_i(e)], s_i'[x \mapsto s_i'(e)]) .
  \]
  \cref{l:comp-marg-supp} shows that $\mu_L \triangleq \dbind(\mu_L', w_L)$
  and $\mu_R \triangleq \dbind(\mu_R', w_R)$ have support in $\domain{\Psi}$ and
  have marginals $\pi_1(\mu_L) = \sem{\sigma; x \gets \expr} s_0$ and $\pi_2(\mu_R)
  = \sem{\sigma; x \gets \expr} s_0'$. It only remains to check the distance
  conditions. By the distance condition on $\mu_L', \mu_R'$, and definition of
  $\overline{f}$, we have
  \begin{align*}
    \mu_L(s, s')
    &= \mu_L'(\{ (t, t') \in \domain{\domain{\Theta}} \mid (s, s') = \sem{x \gets \expr} (t, t') \}) \\
    &\leq \sum_{ (t, t') \in \domain{\Theta} : (s, s') = \sem{x \gets \expr} (t, t') }
    \exp(\overline{f}(t, t')) \cdot \mu_R'(t, t') \\
    &= \exp(f(s, s')) \cdot \mu_R'(\{ (t, t') \in \domain{\Theta} \mid (s, s') = \sem{x \gets \expr} (t, t') \}) \\
    &\leq \exp(f(s, s')) \cdot \mu_R'(\{ (t, t') \in \domain{\Theta} \mid (s, s') = \sem{x \gets \expr} (t, t') \}) \\
    &= \exp(f(s, s')) \cdot \mu_R(s, s') .
  \end{align*}
  Hence $\mu_L, \mu_R$ are witnesses to the desired approximate coupling
  \[
    \sem{\sigma'; x \gets \expr} s_0 \aspace{\Psi_f} \sem{\sigma'; x \gets \expr} s_0' .
  \]

  If the final statement $\stmt$ is an assume $\assume(\bexpr)$,  we know that
  $\domain{\Theta} \subseteq \{ (s, s') \mid s(\bexpr) = s'(\bexpr) \}$ since $\strp$ is a
  synchronizing strategy.  Define a cost function $\overline{f}$ by
  \[
    \overline{f}(s_i, s_i') =
    \begin{cases}
      f(s_i, s_i') &: s_i(\bexpr) \land s_i'(\bexpr) \\
      0 &: \text{otherwise} .
    \end{cases}
  \]
  We have $c \leq \overline{f}(s_i, s_i')$ for every $(s_i, s_i', c) \in
  \Theta$, by definition. By induction, we have an approximate coupling
  \[
    \sem{\sigma'} s_0 \aspace{\Psi_{\overline{f}}} \sem{\sigma'} s_0' .
  \]
  So every pair of inputs $(s_0, s_0') \in \Delta$ gives a pair of witness
  distributions $\mu_L', \mu_R'$.  Now take any two states $(s_i, s_i') \in
  \domain{\Theta}$ and define
  \[
    w_L(s_i, s_i') =
    \begin{cases}
      \dunit(s_i, s_i') &: s_i(\bexpr) \\
      \dzero            &: \text{otherwise}
    \end{cases}
    \quad\text{and}\quad
    w_R(s_i, s_i') =
    \begin{cases}
      \dunit(s_i, s_i') &: s_i'(\bexpr) \\
      \dzero            &: \text{otherwise} .
    \end{cases}
  \]
  \cref{l:comp-marg-supp} shows that $\mu_L \triangleq \dbind(\mu_L', w_L)$
  and $\mu_R \triangleq \dbind(\mu_R', w_R)$ have support in $\domain{\Psi}$ and
  have marginals $\pi_1(\mu_L) = \sem{\sigma'; \assume(\bexpr)} s_0$ and
  $\pi_2(\mu_R) = \sem{\sigma'; \assume(\bexpr)} s_0'$. We just need to check
  the distance condition. By the distance condition on $\mu_L', \mu_R'$, we have
  \begin{align*}
    \mu_L(s, s')
    &= \mu_L'(\{ (t, t') \in \domain{\Theta} \mid t(\bexpr) \}) \\
    &= \mu_L'(\{ (t, t') \in \domain{\Theta} \mid t(\bexpr) \land t'(\bexpr) \}) \\
    &\leq \sum_{ (t, t') \in \domain{\Theta} : t(\bexpr) \land t'(\bexpr) }
    \exp(\overline{f}(t, t')) \cdot \mu_R'(t, t') \\
    &= \exp(f(s, s')) \cdot \mu_R'(\{ (t, t') \in \domain{\Theta} \mid t(\bexpr) \land t'(\bexpr) \}) \\
    &= \exp(f(s, s')) \cdot \mu_R'(\{ (t, t') \in \domain{\Theta} \mid t'(\bexpr) \}) \\
    &\leq \exp(f(s, s')) \cdot \mu_R'(\{ (t, t') \in \domain{\Theta} \mid t'(\bexpr) \}) \\
    &= \exp(f(s, s')) \cdot \mu_R(s, s') .
  \end{align*}
  Hence, $\mu_L, \mu_R$ are witnesses to the desired approximate coupling:
  \[
    \sem{\sigma'; \assume(\bexpr)} s_0 \aspace{\Psi_f} \sem{\sigma'; \assume(\bexpr)} s_0' .
  \]

  If the final statement $\stmt$ is a sampling $x \sim \dexpr$, suppose that
  $\strp$ selects a coupling $\Lambda \subseteq \mathds{Z} \times \mathds{Z}
  \times \mathds{R}$ for the distributions. So for every pair of states $(s,
  s')$ we have
  \[
    s(\dexpr) \aspace{\Lambda} s'(\dexpr) .
  \]
  Note that $\Lambda$ can depend on the states $(s, s')$; we will write
  $\Lambda[s, s']$ to emphasize this.  Define a cost function $\overline{f}$ by
  \[
    \overline{f}(s_i, s_i') \triangleq
    \inf \{ f(s_i[x \mapsto a], s_i[x \mapsto a']) - c
            \mid (a, a', c) \in \Lambda[s_i, s_i'] \}
  \]
  for every $(s_i, s_i') \in \domain{\Theta}$, and $0$ otherwise.
  To give a coupling for the sampling command, we first extend the coupling
  $\Lambda[s, s']$ to whole memories by defining
  \[
    \overline{\Lambda}[s, s'] \triangleq
    \{ (s[x \mapsto a], s'[x \mapsto a'], c) \mid (a, a', c) \in \Lambda \}
  \]
  We then define a cost function $\overline{g}$ representing the cost of the
  coupling: for every $(s_i, s_i') \in \domain{\Theta}$ define
  \[
    \overline{g}(s_i[x \mapsto a], s_i'[x \mapsto a'])
    \triangleq \sup \{ c \mid (a, a', c) \in \Lambda[s_i, s_i'] \}
  \]
  for every $(a, a', c) \in \Lambda[s_i, s_i']$, and $0$ otherwise. Note that
  $\overline{g}$ has a dependence on $(s_i, s_i')$; we will sometimes write
  $\overline{g}[s_i, s_i']$ for clarity. Note also that the supremum on the
  right is at most $f(s_i[x \mapsto a], s_i'[x \mapsto a']) < \infty$. By the
  semantics of the sampling command, we have a coupling
  \[
    \sem{x \sim \dexpr} s_i \aspace{\overline{\Lambda}_{\overline{g}}} \sem{x \sim \dexpr} s_i' .
  \]
  for every $(s_i, s_i') \in \domain{\Theta}$.

  Additionally, $c \leq \overline{f}(s_i, s_i')$ for every $(s_i, s_i', c) \in
  \Theta$ by definition. Hence by induction, we have an approximate coupling
  \[
    \sem{\sigma'} s_0 \aspace{\Theta_{\overline{f}}} \sem{\sigma'} s_0' .
  \]
  So every pair of inputs $(s_0, s_0') \in \Delta$ gives a pair of witness
  distributions $\mu_L', \mu_R'$, in particular they satisfy the marginal and
  support conditions.

  By the definition, $\Theta$ contains all pairs of coupled states reachable
  from $(s_0, s_0')$ after running $\sigma'$, and $\Psi$ contains all pairs of
  coupled states reachable after running $\sigma' ; x \sim \dexpr$. Hence by
  definition of $\cpost_\strp$, we must have $\domain{\overline{\Lambda}[s_i,
    s_i']} \subseteq \domain{\Psi}$ for every $(s_i, s_i') \in \Theta$ and so:
  \[
    \sem{x \sim \dexpr} s_i \aspace{\Psi_{\overline{g}}} \sem{x \sim \dexpr} s_i' .
  \]
  Let $w_L(s_i, s_i'), w_R(s_i, s_i')$ be the left and the right witnesses,
  respectively. \cref{l:comp-marg-supp} shows that $\mu_L \triangleq
  \dbind(\mu_L', w_L)$ and $\mu_R \triangleq \dbind(\mu_R', w_R)$ have support
  in $\domain{\Psi}$ and have the correct marginals $\pi_1(\mu_L) =
  \sem{\sigma'; x \sim \dexpr} s_0$ and $\pi_2(\mu_R) = \sem{\sigma'; x \sim
    \dexpr} s_0'$.

  It only remains to check the distance conditions, and it is enough to check
  for final states $(s, s') \in \Psi$---otherwise by the support condition,
  $\mu_L(s, s') = 0$ and the distance condition is clear. By the
  distance condition on $\mu_L', \mu_R'$, we have:
  \begin{align*}
    \mu_L(s, s')
    &= \sum_{(s_i, s_i') \in \domain{\Theta}}
    \mu_L'(s_i, s_i') \cdot w_L(s_i, s_i')(s, s') \\
    &\leq \sum_{(s_i, s_i') \in \domain{\Theta}}
    \exp(\overline{f}(s_i, s_i') + \overline{g}[s_i, s_i'](s, s'))
    \cdot \mu_R'(s_i, s_i') \cdot w_R(s_i, s_i')(s, s') \\
    &\leq \sum_{(s_i, s_i') \in \domain{\Theta}}
    \exp(f(s, s') - \overline{g}[s_i, s_i'](s, s') + \overline{g}[s_i, s_i'](s, s'))
    \cdot \mu_R'(s_i, s_i') \cdot w_R(s_i, s_i')(s, s') \\
    &\leq \exp(f(s, s')) \sum_{(s_i, s_i') \in \domain{\Theta}} \mu_R'(s_i, s_i') \cdot w_R(s_i, s_i')(s, s') \\
    &= \exp(f(s, s')) \cdot \mu_R(s, s') .
  \end{align*}
  The second inequality holds because we may assume that $(s, s')$ are of the
  form $(s_i[x \mapsto a], s_i'[x \mapsto a'])$ with $(a, a') \in \Lambda[s_i,
  s_i']$; if not, $w_R(s_i, s_i')(s, s') = 0$ by the support condition on
  witnesses. Hence, we have the claimed approximate coupling:
  \[
    \sem{\sigma'; x \sim \dexpr} s_0
    \aspace{\Psi_f}
    \sem{\sigma'; x \sim \dexpr} s_0' .
  \]
  This completes the inductive step, and the proof.
\end{proof}

\subsection{Proofs for HMC encoding}
First, we define how we interpret a solution $\interp \models \clss$
 as a strategy $\{\strp_j\}_j$.
 For every pair of states $(s_1,s_2) \in S \times S$,
 every sampling statement $\stmt_e$,
 and every $j \in \dom_\retv$,
 we define $\strp_j$ as:
 \[
 \strp_j(\stmt_e, s_1, s_2) = \interp(\strategy_e(\vec{c},-,-,-))
 \]
 where $\vec{c}$ assigns each variable $v_i$ in $\vec{v}$
 the value $s_i(v)$, and assigns $\ret$ to $j$.
 By semantics of coupling constraints,
 any solution $\interp \models \clss$ results in a coupling strategy
 using the above construction

\begin{lemma}[Soundness of coupled invariants]\label{lemma:hmc}
  Fix $P$ and $\adj$.
  Let $\clss$ be the generated Horn clauses.
  Let $\interp \models \clss$.
  Let $\{\strp_j\}_j$ be the
  coupling strategy in $\interp$, as defined above.
  For every location $i$ in $P$,
   every trace $\trace$ from $\lentry$ to $i$,
   every $j \in \dom_\retv$,
  and every $(s_1,s_2,d) \in \cpost_{\strp_j}(\adj \times \{0\}, \trace)$,
  we have that the following formula is valid
  \[
  \left(\costvar = d \land \ret = j \land \bigwedge_{v \in \vars} s_1(v) = v_1 \land s_2(v) = v_2\right)
  \Rightarrow \interp(\inv_i(\vec{v},\costvar))
  \]
\end{lemma}

\begin{proof}
We proceed by induction on length $n$ of $\trace$.

\textbf{Base case:}
For $n = 0$ (empty trace),
the constraint above holds
trivially by rule $\cinit$.

For $n = 1$,
we consider 3 different cases, one for each statement type.

Assume $\trace$ is an assignment statement $v \gets \expr$,
then we know that $(s_1, s_2, d) \in \adj \times \{0\}$
iff $(s_1[v \mapsto s_1(\expr)], s_2[v \mapsto s_2(\expr)],d) \in \cpost(\adj \times \{0\}, \trace)$.
By definition of $\cassign$,
we know that
the following is valid:
 \[\interp(\inv_\lentry(\vec{v},\costvar))
\land v_1' = \expr_1 \land v_2' = \expr_2 \longrightarrow
\interp(\inv_i(\vec{v}[v],\costvar))\]
By base case $n=0$,
we know that the following is valid:
\[
\left(\costvar = d \land \ret = j \land \bigwedge_{v \in \vars} s_1(v) = v_1 \land s_2(v) = v_2\right)
\Rightarrow \interp(\inv_\lentry(\vec{v},\costvar))
\]
Therefore,
\[
\left(\costvar = d \land \ret = j \land \bigwedge_{v \in \vars} s_1[v \mapsto s_1(\expr)](v) = v_1 \land s_2[v \mapsto s_2(\expr)](v) = v_2\right)
\Rightarrow \interp(\inv_i(\vec{v},\costvar))
\]

Now, assume $\trace$ is $\texttt{assume}(\bexpr)$.
By definition of $\cpost$, for every
$(s_1, s_2, d) \in \adj \times \{0\}$,
$(s_1, s_2, d) \in \cpost(\adj \times \{0\}, \texttt{assume}(\bexpr))$
iff $s_1(\bexpr) = s_2(\bexpr) = \true$.
By $n=0$ and the definition of rule $\casmconst$,
we have that the following is valid:
\[
\left(\costvar = d \land \ret = j \land \bigwedge_{v \in \vars} s_1(v) = v_1 \land s_2(v) = v_2\right) \land \bexpr_1 \land \bexpr_2
\Rightarrow \interp(\inv_i(\vec{v},\costvar))
\]
Therefore, if $(s_1,s_2,d)$ is such that $s_1(\bexpr) = s_2(\bexpr) = \true$,
then the left hand side of the above formula is satisfied.

Now, assume $\trace$ is $v \sim \dexpr$.
By definition of $\cpost$,
we know that for every $(s_1,s_2,d) \in \adj \times 0$,
and every $(c_1,c_2,d') \in \strp_{j}(v \sim \dexpr, s_1,s_2)$,
we have $(s_1[v \mapsto c_1] , s_2[v \mapsto c_2], d + d') \in \cpost(\adj \times 0, \trace)$.
By definition of $\strp_j$ using $\interp$ (as defined above)
and the definition of rule $\cstrategy$,
we have that the following holds:

\[
\left(\costvar = d + d' \land \ret = j \land \bigwedge_{v \in \vars} s_1[v \mapsto c_1](v) = v_1 \land s_2[v \mapsto c_2](v) = v_2\right)
\Rightarrow \interp(\inv_i(\vec{v},\costvar))
\]

\textbf{Inductive step:}
Assuming the lemma holds for $n > 1$,
then establishing that it holds for $n+1$
follows an analogous case split by statement as for the base case.
\end{proof}

We can now easily prove \cref{thm:hmc}.
\begin{proof}
  Fix $P$ and $\adj$.
  Let $\clss$ be the generated Horn clauses.
  Let $\interp \models \clss$.
  Let $\{\strp_j\}_j$ be the
  coupling strategy in $\interp(\strategy_e)$,
  for every $e$ where $\stmt_e$ is a sampling statement.
  For
   every maximal trace $\trace$,
   every $j \in \dom_\retv$,
  and every $(s_1,s_2,d) \in \cpost_{\strp_j}(\adj \times \{0\}, \trace)$,
  we know from \cref{lemma:hmc} that
  \[
  \left(\costvar = d \land \ret = j \land \bigwedge_{v \in \vars} s_1(v) = v_1 \land s_2(v) = v_2\right)
  \Rightarrow \interp(\inv_\lexit(\vec{v},\costvar))
  \]
  We also know by rule \cdp{} that
  $\interp(\inv_\lexit(\vec{v},\costvar)) \longrightarrow
  \costvar \leq \eps \land (\retva = \ret \Rightarrow \retvb = \ret)$
  is valid.
Therefore,
$d\leq \eps$ and, if $s_1(\retv) =j$, then $s_2(\retv) = j$.

This establishes \edp, assuming $\{\strp_j\}_j$
are synchronizing.
We now show that strategy is synchronizing.
Fix some statement $\texttt{assume}(\bexpr)$
on edge $(i,j)$.
By rule $\casmsync$, we have that
the following is valid:
$\interp(\inv_i(\vec{v},\costvar)) \to \bexpr_1 \equiv \bexpr_2$.
Therefore, following \cref{lemma:hmc}, for any trace $\trace$
from $\lentry$ to $i$,
every $j$ and every  $(s_1,s_2,d) \in \cpost_{\strp_j}(\adj \times \{0\}, \trace)$,
we have $s_1(\bexpr) = s_2(\bexpr)$.
\end{proof}

We now prove correctness of \cref{thm:trans},
which states correctness of transforming coupling
constraints into clauses.

\begin{proof}
  Fix a coupling constraint in $\clss$:
  \[
    M(\vec{v}_1) \aspace{\rel(\vec{v}_3,-,-,-)} M(\vec{v}_2)
  \]
  Let $\interp \models \clss'$.
  Suppose that $\interp$ does not satisfy the coupling constraint.
  Then, there is a value of $\vec{c}_{1,2,3}$ of $\cvars_{1,2,3}$
  such that
  $\interp{(\rel(\vec{c}_3,-,-,-))}$ is not a coupling
  for distributions $M(\vec{c}_1)$ and $M(\vec{c}_2)$.

  By \cref{def:trans}, we know that $\interp(\rel(\cvars,d_1,d_2,\theta))$
  is a formula of the form
\[
\left(\bigwedge_{i=1}^n f(\cvars) = i \implies \trans(\Lambda^i[\cvars_1,\cvars_2])\right)  \]
Then, for some $i \in [1,n]$,
the relation over $(d_1,d_2,\theta)$
denoted by the formula $\Lambda[\vec{c}_1,\vec{c}_2]$
is \rone not a coupling for $M(\vec{c}_1)$ and $M(\vec{c}_2)$,
or \rtwo $\pre_{\Lambda^i}[\vec{c}_1,\vec{c}_2]$
does not hold.
To show that neither case is true,
we prove correctness of $\trans$ in \cref{fig:enc}.
We consider the following cases
of $\Lambda^i$:
\begin{itemize}
  \item \textbf{Null $\Lambda_\varnothing[z_1,z_2]$:}
  For every instantiation $c_1,c_2$ of
  $z_1,z_2$, we have the formula $d_1 - c_1 = d_2 - c_2 \land \theta = 0$
  whose satisfying assignments are the ternary relation
  $\{(d_1, d_2, 0) \mid d_1,d_2 \in \mathds{Z}, d_1 -c_1 = d_2 - c_2\}$,
  recovering the definition of a null coupling.

  \item \textbf{Shift (Lap) $\Lambda_{+k}[z_1,z_2]$:}
  For every instantiation $c_1,c_2$ of
  $z_1,z_2$, we have the formula $d_1 + \interp(f_k) = d_2 \land \theta = |c_1-c_2 +\interp(f_k)|\eps$, where $\interp(f_k)$ is an integer.
  The satisfying assignments of this formula
  are
  \[
  \{(d_1, d_2, |c_1-c_2 +\interp(f_k)|\eps) \mid d_1 + \interp(f_k) = d_2, d_1,d_2 \in \mathds{Z}\}
  \]
  This recovers the definition of a shift coupling for $\lap$.

  \item \textbf{Shift (Exp) $\Lambda_{+k}'[z_1,z_2]$:}
  For every instantiation $c_1,c_2$ of
  $z_1,z_2$ such that $c_2 - c_1 \leq \interp(f_k)$,
  we have a similar argument as above for $\lap$.
  By the precondition constraint specified
  in \cref{def:trans},
  we know that if $f(\vec{v}) = i$, then
  $z_2 - z_1 \leq k$.

  \item \textbf{Choice $(\Pi \mathrel{?} \Lambda_{+k}[z_1,z_2]: \Lambda_\varnothing[z_1,z_2])$:}
  For every instantiation $c_1,c_2$ of $z_1,z_2$,
  we have the formula
  \[
    \bigwedge
    \left\{
      \begin{array}{ll}
        \interp(\Pi(d_1) &\Rightarrow  (d_1 - c_1 = d_2 - c_2 \land \theta = 0)
        \\
        \neg \interp(\Pi(d_1)) &\Rightarrow (d_1 + \interp(f_k) = d_2 \land \theta = |c_1-c_2 +\interp(f_k)|\eps)
      \end{array}
    \right.
    \]
    which according to our cases above,
    and \cref{lem:combo-couple},
    is a coupling assuming the non-overlapping condition.
    By the precondition constraint specified
    in \cref{def:trans},
    we know that if $f(\vec{v}) = i$, then
    for every $d_1,d_2,d_1',d_2'$,
    the following formula holds:

    \[
      \bigwedge
      \left\{
        \begin{array}{ll}
          \interp(\Pi(d_1)) &\Rightarrow d_1 + \interp(f_k) = d_2 \land \theta = |c_1-c_2 +\interp(f_k)|\eps
          \\
          \neg \interp(\Pi(d_1')) &\Rightarrow d_1' - c_1 = d_2' - c_2 \land \theta' = 0
        \end{array}
      \right\}
      \Longrightarrow
      d_2' \neq d_2
    \]
    The first conjunct encodes the shift coupling (as shown above)
    and the second conjunct the null coupling (also shown above).

  This implies the non-overlapping condition:
  for $a_1 \in \Pi$, $a_1' \not\in \Pi$,
   every $(a_1, a_2, -) \in \Lambda_{+k}[c_1,c_2]$ and
   every $(a_1', a_2', -) \in \Lambda_{+k}[c_1,c_2]$,
   we have $a_2 \neq a_2'$.
   \qedhere
\end{itemize}
\end{proof}

\section{Implementation details}

\subsection{Dealing with non-linearity}
In \cref{sec:evaluation}, we noted that our tool employs an optimization wherein
it fixes the parameter $\eps$ to 1, and that suffices
to show differential privacy for all $\eps > 0$.
We now formally show correctness of this optimization.

Our optimization hinges on three observations:
\rone In all algorithms in the literature,
$\eps$ only appears in the form $\iexpr \cdot\eps$
in sampling statements, e.g., $\lap_{\iexpr\cdot\eps}(\cdot)$
or $\olap_{\iexpr\cdot\eps}(\cdot)$.
Here, $\iexpr > 0$ is assumed to be a positive expression
over input program variables.
\rtwo In all couplings we utilize,
the cost is a multiple of $\iexpr\cdot\eps$.
\rthree The choice of a coupling family and its parameters
in our coupling strategies does not depend
on the value of $\eps$.
Therefore, for any coupled
execution from a pair of adjacent states,
and for any $\eps$,
the final cost is of the form $\sum_{i=1}^n c_i\eps$,
where there are $n$ sampling statements encountered in the coupled
execution and $c_i \in \mathds{R}^{\geq 0}$.
Since our goal is to show that $\sum_{i=1}^n c_i\eps \leq \eps$,
it suffices to check that $\sum_{i=1}^n c_i \leq 1$.

The following Lemma formalizes
the our observations,
under the assumptions stated above.

\begin{lemma}\label{lemma:nonlinear}
  Fix program $\prog$ and adjacency relation $\adj$.
  Let $\trace$ be a maximal trace through $\prog$
  and $\strp$ be coupling strategy for $P$.
  Let $Q_c = \cpost_\strp(\adj_c \times \{0\}, \trace)$,
  for $\adj_c$ is $\adj$ with
   $\eps$ fixed to  $c \in \mathds{R}^{\geq 0}$
   in all pairs of states.
  Then,
  \begin{enumerate}
  \item for all $c>0$,
   $(s_1,s_2,-) \in Q_c$
   iff $(s_1,s_2,-) \in Q_1$;
  \item for all $c > 0$,
  if $(s_1,s_2,d) \in Q_c$,
  then $(s_1,s_2,d/c) \in Q_1$; and
  \item for all $c> 0$, if $(s_1,s_2,d) \in Q_1$,
  then $(s_1,s_2,d\cdot c) \in Q_c$.
\end{enumerate}
\end{lemma}

Point 1 states how states are paired does not change
by changing $\eps$.
Points 2 and 3 state that the final cost
for a paired execution is a multiple of $\eps$.

\begin{proof}
\emph{Point 1} immediately follows from the fact
that the choice of coupling in $\tau$ does not
depend on $\eps$, and that $\eps$ does not
appear other than in sampling statements.
This ensures that all paired states in $Q_c$ and $Q_1$
are the same.

\emph{Point 2}
Suppose that $(s_1,s_2,d) \in Q_c$.
By definition of $\cpost_\tau$,
shift and null couplings,
and our assumption that all sampling statements
are of the form $\lap_\iexpr\eps(-)$ or $\olap_{\iexpr\eps}(-)$,
then we know that $d$ is of the form
$\sum_{i=1}^n a_i c$,
where $a_i \in \mathds{Z}$.
By Point 1 and the assumption that $\eps$ does not affect
the choice of coupling,
we know that $(s_1,s_2,\sum_{i=1}^n a_i) \in Q_1$.

\emph{Point 3} follows analogously.
\end{proof}

The following theorem immediately follows from the above lemma
and \cref{thm:winning};
it shows that if we find a winning coupling strategy
for $\eps = 1$, then we prove \edp.

\begin{theorem}\label{thm:nonlinear}
  Fix program $\prog$ and adjacency relation $\adj$.
  Suppose that we have a family of synchronizing coupling strategies $\{ \strp_j
  \}_{j \in \dom_\retv}$ such that for every trace $\trace \in \traces(P)$, we
  have
  \[
  \cpost_{\strp_j} (\adj_1 \times \{0\},\trace) \subseteq
  \{(s_1',s_2',c) \mid c \leq 1 \land s_1'(\retv) = j \Rightarrow s_2'(\retv) = j\},
  \]
  where $\adj_1$ is as defined in \cref{lemma:nonlinear}.
  Then, $\prog$ is $\eps$-differentially private.
\end{theorem}

\subsection{Algorithmic details}
We now present more details details on
the algorithm we used in our implementation
for solving Horn clauses with uninterpreted functions.
The functions \textsc{synth} and \textsc{verify} are treated
as black-box external solvers that we assume we have access to.
The primary piece we describe here is the unrolling
of clauses $\clss$ into a non-recursive set $\clss^n$,
and encoding them for the synthesizer $\textsc{synth}$.

\paragraph{Unrolling.}
The unrolling/unwinding process of clauses
$\clss$ to $\clss^n$ is a standard
procedure that follows, for example,
the algorithm of~\citet{mcmillan2013solving}.
We assume there is a single clause $Q \in \clss$
whose head is not a relation application---we call it
the \emph{query} clause.
The unrolling starts with clause $Q$,
and creates a fresh copy of every clause
whose head is a relation that appears in the body of $Q$,
and so on, recursively, up to depth $n$.
This is analogous to unrolling program loops $n$ times.

The procedure \textsc{unroll} is shown
in \cref{alg:unroll}.
We call it with $\textsc{unroll}(\clss,Q,n)$,
where $n > 0$.
We use $\cls^n$ to denote the clause $\cls$
but with head relation $r$ tagged with number $n+1$, i.e.,  $r^{n+1}$,
and with body relations tagged with number $n$.
These are new relation symbols resulting from the unrolling process.

\begin{algorithm}
  \begin{algorithmic}
    \Function{unroll}{$\clss, \cls, n$}
      \If {$n=0$} \Return $\emptyset$ \EndIf
      \State let $B = \{ r \mid r(\vec{v}) \text{ in body of } \cls\}$
      \State let $\mathcal{U} = \{\cls \mid \cls \in \clss \text{ with relation in head is in set } B\}$
      \State let $\clss' = \{\cls^n\} \cup \bigcup_{U\in \mathcal{U}} \textsc{unroll}(\clss,U,n-1)$
      \State \Return $C'$
    \EndFunction
  \end{algorithmic}
  \caption{Unrolling procedure}\label{alg:unroll}
\end{algorithm}

\begin{example}
  Consider the following clauses:
  \begin{align*}
    r(x) \longrightarrow x > 0\\
    r(x) \land f(x) = x' \longrightarrow r(x')
  \end{align*}
  If we unroll these clauses up to $n = 3$, we get the following:
  \begin{align*}
    r^3(x) \longrightarrow x > 0\\
    r^2(x) \land f(x) = x' \longrightarrow r^3(x')\\
    r^1(x) \land f(x) = x' \longrightarrow r^2(x')\\
  \end{align*}
  Notice that the clauses resulting from unrolling are non-recursive.
\end{example}

\paragraph{Synthesis.}
Once we have non-recursive clauses $\clss$,
we can construct a formula of the form
$$\exists f_1, \ldots, f_n \ldotp \exists r_1, \ldots, r_m \ldotp \bigwedge_{C \in \clss} \cls$$
where
 $f_i$ and $r_i$ are the function and relation symbols appearing in the unrolled clauses.

Since $\clss$ are non-recursive,
we can rewrite them to remove the relation
symbols.
This is a standard encoding in Horn clause solving
 that is analogous
to encoding a loop-free program as a formula,
as in VC generation and bounded model checking,
and we do not detail it here.
We thus transform $\bigwedge_{C \in \clss} \cls$
into a formula of the form $\forall \vec{x} \ldotp \varphi$.
(Recall that $\clss$ are universally quantified.)

At this point, we have a problem of the form
$$\exists f_1, \ldots, f_n \ldotp \forall \vec{x} \ldotp \varphi$$
which can be solved using program synthesis algorithms.

\paragraph{Soundness.}
The soundness of the algorithm \textsc{solve}, in \cref{ssec:solving},
follows trivially from the soundness of the \textsc{verify}
procedure.
In other words, no matter what \textsc{synth}
returns, assuming \textsc{verify} is sound,
\textsc{solve} will be sound.
This is formalized in the following theorem.

\begin{theorem}
  Let $\clss$ be a set of Horn clauses
  with potentially uninterpreted functions.
  Assume that, given a set of clauses $\clss'$
  with no uninterpreted functions, if $\textsc{verify}(\clss')$
  return \abr{SAT}, then $\clss'$ is satisfiable.
  Then, if $\textsc{solve}(\clss)$ returns \abr{SAT},
  $\clss$ are satisfiable.
\end{theorem}
% We do so by encoding the unrolled clauses as
% a formula with no uninterpreted relations, similar to
% what standard Horn-clause solvers when, for example,
% they construct a formula for interpolation.

% This process is performed by the function
% \textsc{flatten}, which is called
% on the set of unrolled clauses $\clss$
% and the only query clause $Q$:
% \textsc{flatten}$(\clss,Q)$.
%
% \begin{algorithm}
%   \begin{algorithmic}
%     \Function{flatten}{$\clss, \cls$}
%       \State let $B = \{ r \mid r(\vec{v}) \text{ in body of } \cls\}$
%       \For {$\emph{body} \to r(\vec{v}') \in B$}
%         \State assign a fresh Boolean variable $b_i$ to each $r_i$ in \emph{body}.
%       \EndFor
%     \EndFunction
%   \end{algorithmic}
%   \caption{Unrolling procedure}\label{alg:unroll}
% \end{algorithm}

%!TEX root=paper.tex
}{}

\end{document}